\documentclass[aps,prl,reprint,amsmath,amssymb,longbibliography]{revtex4-2}
\usepackage{graphicx,epsfig,epsf,color,hhline,import}
\usepackage{dcolumn}
\usepackage{bm}
\usepackage{tensor,pbox}
\usepackage{tikz}
\usepackage{epstopdf}
\usepackage{amsthm}
\usepackage{bbold}

\usepackage{changes}
\usepackage{cancel}
\usepackage{color}
\usepackage{url}
\usepackage[hidelinks]{hyperref}
\usepackage{ytableau}
\usepackage{enumerate}
\usepackage{mathdots}
\usepackage{braket}
\usepackage{MnSymbol}
\usepackage{mathtools}

\DeclarePairedDelimiter\ceil{\lceil}{\rceil}
\DeclarePairedDelimiter\floor{\lfloor}{\rfloor}

\newtheorem{theorem}{Theorem}

\newtheorem{lemma}[theorem]{Lemma}
\newtheorem{proposition}[theorem]{Proposition}
\newtheorem{definition}[theorem]{Definition}

\newcommand{\tr}{{\rm tr}}

\newcommand{\trb}[1]{\text{tr}\left[{#1}\right]}

\newcommand{\mc}[1]{\mathcal{#1}}

\newcommand{\mm}[1]{\mathrm{#1}}

\newcommand{\pij}[2]{p_{#1,#2}} 
\newcommand{\qij}[2]{q_{#1,#2}} 
\newcommand{\param}{\ensuremath{\alpha}} 
\newcommand{\dgen}{\ensuremath{h_\param}} 
\newcommand{\pnorm}[2]{\ensuremath{ \left\Vert  #1  \right\Vert_{#2} }}

\newcommand{\qfinkd}{\ensuremath{\mathop{}\! I}} 

\newcommand{\ii}{\ensuremath{ \mathrm{i\,} }}
\newcommand{\ens}[1]{\ensuremath{ {\{ #1 \}}} } 
\newcommand{\absv}[1]{\ensuremath{ \left\vert   #1 \right\vert }}  

\newcommand{\diff}[1][]{\mathop{}\!\mathrm{d_{#1}}}
\newcommand{\ee}{\mathrm{e}}

\definecolor{bostonuniversityred}{rgb}{0.8, 0.0, 0.0}

\begin{abstract}
	We study quantum metrology for unitary dynamics. Analytic solutions
	are given for both
	the optimal unitary state preparation starting
	from an arbitrary mixed state and the corresponding optimal
	measurement precision. This represents a rigorous generalization of
	known results for optimal initial states and upper bounds on
	measurement precision which can only be saturated if pure states are
	available. In particular, we provide a generalization to mixed
	states
	of an upper bound on measurement precision for time-dependent
	Hamiltonians that can
	be saturated with optimal
	Hamiltonian control. These results make precise and reveal the full potential of mixed states for quantum metrology.
\end{abstract}
\begin{document}
	\title{Maximal Quantum Fisher Information for Mixed States}
	\author{Lukas J.~Fiderer$^1$, Julien M.\,E. Fra\"isse$^2$, Daniel Braun$^1$}
	\affiliation{$^1$Eberhard-Karls-Universit\"at T\"ubingen, Institut f\"ur Theoretische Physik, 72076 T\"ubingen, Germany\\
		$^2$Seoul National University, Department of Physics and Astronomy, Center for Theoretical Physics, 151-747 Seoul, Korea}
	\maketitle

	The standard paradigm of quantum metrology involves the preparation of an initial state, a parameter-dependent dynamics, and a consecutive quantum measurement of the evolved state.
	From the measurement outcomes the parameter can be estimated
	\cite{giovannetti_quantum_2004,paris_quantum_2009,giovannetti_advances_2011}. Naturally,
	it is the goal to estimate the parameter as precisely as possible,
	i.e., to reduce the uncertainty
	$\Delta\hat{\alpha}=\text{Var}(\hat{\alpha})^{1/2}$ of the estimator
	$\hat{\alpha}$ 
	of the parameter $\alpha$ that
	we want to
	estimate. We consider single parameter estimation in the local regime where one already has a good estimate $\hat{\alpha}$  at hand (typically from prior measurements) such that this prior knowledge can be used to prepare and control consecutive measurements. Quantum coherence and non-classical correlations in quantum sensors help to reduce the uncertainty $\Delta\hat{\alpha}$ compared to what is possible with comparable classical resources \cite{pezze2018quantum,braun2018quantum}.
	The ultimate precision limit for unbiased estimators is given by the quantum Cram\'er--Rao bound $\Delta\hat{\alpha}\geq (M\qfinkd_\param)^{-1/2}$ which depends on the number of measurements $M$ and the quantum Fisher information (QFI) $\qfinkd_\param$ which is a function of the
	state  \cite{helstrom_quantum_1976,braunstein_statistical_1994}.
	When the number of measurements is fixed, as they correspond to a
	limited resource, precision is optimal and the QFI is maximal which
	involves an optimization with respect to the  
	state. 
	
	In this Letter, we consider a freely available state $\rho$, unitary freedom to prepare an initial state from $\rho$, and unitary parameter-dependent dynamics of the quantum system {(see Fig.~\ref{fig:protocol})}.
	The parameter-dependent dynamics will be called sensor dynamics in the following in order to distinguish it from the state preparation dynamics. For instance, in a spin system the unitary freedom can be used to squeeze the spin before it is subjected to the sensor dynamics, as it is the case in many quantum-enhanced measurements  \cite{fernholz2008spin,andre2002atom,leroux2010implementation,orzel2001squeezed}. In the worst case scenario, only the maximally mixed state is available, which does not change under unitary state preparation or unitary sensor dynamics and, thus, no information about the parameter can be gained. In the best-case scenario the available state is pure, when
	the
	maximal QFI as well as the optimal state to be prepared 
	are well-known  \cite{giovannetti_quantum_2006,fujiwara_fibre_2008}. 
	
	The appeal and advantage of the theoretical study of unitary sensor
	dynamics lies in the analytic solutions that can be found that allow
	fundamental insights in the limits of quantum metrology and the role
	of resources such as measurement time and system size. The QFI
	maximized with respect to initial states, also known as channel QFI,   
	can be reached only with pure initial states. If only mixed states are
	available, as it is usually the case under realistic conditions, this
	upper bound cannot be saturated and therefore has limited
	significance. 
	In fact, if pure states are not available, the question for the 
	maximal QFI and optimal state to be prepared is an
	important open problem 
	\cite{modi2011quantum,haine2015quantum}. 
	The main result of this Letter,
	theorem \ref{th:main} below, 
	is the complete solution of this
	problem.

	The  solution is relevant
	for practically all quantum sensors, as perfect initialization to a
	pure state can only be achieved to a certain degree that varies with
	the quantum system and the available technology. For example, nitrogen-vacancy (NV) center arrays  \cite{pham2011magnetic, barry2016optical} or
	atomic-vapor magnetometers  \cite{savukov2005effects,budker2007optical} operate with mixed initial states due to imperfect
	polarization and competing depolarization
	effects  \cite{appelt1998theory,choi2017depolarization}.
	Particularly
	relevant is the example of sensors based on nuclear spin ensembles that typically
	operate with nuclear spins in thermal equilibrium, such that at room
	temperature the available 
	state is strongly mixed  \cite{gershenfeld1997bulk}. 
	Hence, 
	the full 
	potential of quantum metrology is exploited only when
	the mixedness of initial states is taken into account
	\cite{jones2009magnetic,simmons2010magnetic,schaffry2010quantum,modi2011quantum}.
	
	\begin{figure}
		\centering
		\includegraphics[width=0.49\textwidth]{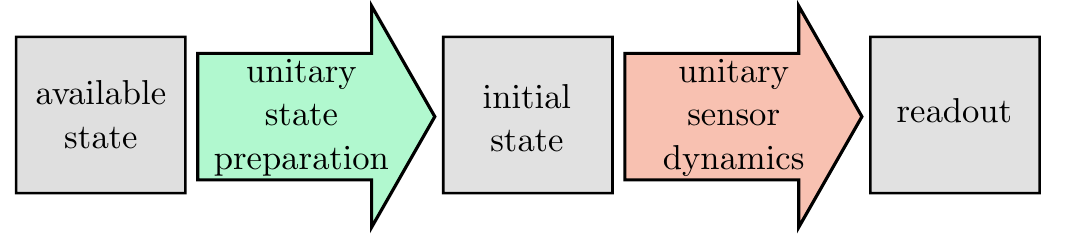}\\
		\caption{Schematic representation of the metrology protocol.}\label{fig:protocol} 
	\end{figure}
	We consider arbitrary, possibly time-dependent Hamiltonians
	$H_\alpha(t)$ for the sensor dynamics. The corresponding unitary
	evolution operator is $U_\param\coloneqq
	\mc{T}(\exp[-\ii/\hbar\int_0^T H_\alpha(t) \diff t])$, 
	where $\mc{T}$ denotes time-ordering, $T$ is the total time of the
	sensor dynamics, and  we set $\hbar=1$ in the following. In the
	simplest case, dynamics is generated by a ``phase-shift'' or
	``precession'' Hamiltonian proportional to the parameter $\alpha$,
	$H_\alpha=\alpha G$, with some parameter-independent operator $G$. The
	parameter dependence of the sensor dynamics is characterized by the
	generator 
	$\dgen\coloneqq \ii U_\alpha^\dagger\frac{\partial U_\alpha}{\partial \alpha}$,
	which simplifies to $G$ for phase-shift Hamiltonians  \cite{giovannetti_quantum_2006,boixo_generalized_2007,pang_quantum_2014,liu2015quantum}.
	
	By introducing the eigendecomposition of the prepared initial state $\rho={\sum_{k=1}^d} p_k\ket{\psi_k}\bra{\psi_k}$, where $d$ is the dimension of the Hilbert space,
	the QFI can be expressed as  \cite{braunstein_statistical_1994}, \cite{modi2011quantum}
	\begin{align}\label{eq:qfi}
	\qfinkd_\param (\rho)\coloneqq 2\sum_{k,\ell=1}^d \pij{k}{\ell}|\bra{\psi_k}\dgen\ket{\psi_\ell}|^2,
	\end{align} 
	with coefficients 
	\begin{align}\label{eq:pij}
	\pij{k}{\ell}\coloneqq \begin{cases*}
	0 \text{~\quad\qquad if }p_k=p_\ell=0,\\
	\frac{(p_k-p_\ell)^2}{p_k+p_\ell}\text{\quad else.}
	\end{cases*}
	\end{align}
	Also,  let $\mm{U}(d)$ denote the set of $d\times d$ unitary matrices.
	\begin{theorem}\label{th:main}
		For any state $\rho$ and any generator $\dgen$ with ordered eigenvalues $p_1\geq \dotsm \geq p_d$ and $h_1\geq\dotsm\geq h_d$, respectively, the maximal QFI with respect to all unitary state preparations $U\rho U^\dagger$, $U\in\emph{U}(d)$,
		is given by
		\begin{align}\label{eq:opt_QFI_th1}
		I_\alpha^*\coloneqq \max_{U} \qfinkd_\param (U\rho U^\dagger)=\frac{1}{2}\sum_{k=1}^{d}\pij{k}{d-k+1}(h_k-h_{d-k+1})^2.
		\end{align}
		Let $\ket{h_k}$ be the eigenvectors of the generator,  $\dgen\ket{h_k}=h_k\ket{h_k}$. The maximum $I_\alpha^*$ is obtained by preparing the initial state
		\begin{equation}\label{eq:opt_state}
		\rho^*=\sum_{k=1}^d p_k \ket{\phi_k}\bra{\phi_k}\,,
		\end{equation}
		with\footnote{More generally, the eigenvectors of $\rho^*$ may be written with a relative phase in the superpositions of $\ket{h_k}$ ($\ket{\mu_k(0)}$ in case of theorem \ref{th:time_dep}) such that the orthonormality condition remains fulfilled, i.e., $\frac{\ket{h_k}+e^{\ii\varphi_k}\ket{h_{d-k+1}}}{\sqrt{2}}\text{ if } 2k<d+1$ and $\frac{\ket{h_k}-e^{\ii\varphi_k}\ket{h_{d-k+1}}}{\sqrt{2}}\text{ if } 2k>d+1$, where the same applies for theorem \ref{th:time_dep} when replacing the $\ket{h_k}$ by $\ket{\mu_k(0)}$.}
		\begin{equation}\label{eq:opt_eigen_th1}
		\ket{\phi_k}=\begin{cases*}
		\frac{\ket{h_k}+\ket{h_{d-k+1}}}{\sqrt{2}}\quad ~~\text{if } 2k<d+1, \\
		\mathrlap{\ket{h_k}}\hphantom{\frac{\ket{h_k}+\ket{h_{d-k+1}}}{\sqrt{2}}}\quad ~~\text{if }2k=d+1,\\
		\frac{\ket{h_k}-\ket{h_{d-k+1}}}{\sqrt{2}}\quad ~~\text{if } 2k>d+1.
		\end{cases*}
		\end{equation}
	\end{theorem}
	The proof is based on the Bloomfield--Watson inequality on
	the Hilbert--Schmidt norm of off-diagonal blocks of a Hermitian matrix  \cite{bloomfield1975inefficiency,drury2002some} and is
	given in the Supplemental Material \footnote{See Supplemental
		Material at [URL will be inserted by publisher] for proofs of
		theorem 1 and theorem 2, as well as the proofs of Heisenberg scaling for thermal states.}. 
	The idea of the proof is to construct an upper bound for the QFI in Eq.~\eqref{eq:opt_QFI_th1} that  exhibits a simpler dependence on the coefficients $\pij{k}{\ell}$. Then we maximize the upper bound by exploiting the Bloomfield--Watson inequality. The proof is concluded by showing that at its maximum the upper bound equals the QFI.

	It is important to notice that the rank $r$ of the state $\rho$ plays
	a crucial role both for the maximal QFI and for the optimal state:
	In order to reach the maximal QFI $I_\alpha^*$, the choice of the
	$\ket{\phi_k}$ corresponding to vanishing $p_k$, i.e., for $k> r$, is
	irrelevant. This is best exemplified by considering the well-known
	case of pure states,  characterized by $p_1=1$ and $r=1$
	\cite{giovannetti_quantum_2006,boixo_generalized_2007,pang_quantum_2014,fraisse_enhancing_2017,pang_optimal_2017}. Then,
	the maximal QFI in Eq.~\eqref{eq:opt_QFI_th1} simply becomes
	$(h_1-h_{d})^2$ and is obtained by preparing an equal superposition $
	(\ket{h_1}+\ket{h_d})/\sqrt{2}$ of the eigenvectors corresponding to
	the minimal and maximal eigenvalues of $\dgen$.  
	When the rank is increased but remains less than or equal to $(d+1)/2$,  the optimal QFI is equal 
	to $\sum_{i=1}^{r}p_i (h_i-h_{d-i+1})^2$. This can be seen as a convex sum of 
	pure-state QFIs \footnote{Yu showed that \textit{any} state $\rho$ can be decomposed as $\sum_k u_k\ket{U_k}\bra{U_k}$ with weights $u_k$ and 
	(generally non-orthonormal) pure states $\ket{U_k}$ such that the QFI equals 
	$\sum_k u_kI_\alpha(\ket{U_k}\bra{U_k})$ \cite{Toth2013Extremal, yu2013quantum}. For $\rho^*$ in Eq.~\eqref{eq:opt_state}, with rank $r\leq (d+1)/2$, Yu's state decomposition of $\rho^*$ is given by the eigendecomposition of $\rho^*$, $\ket{U_k}=\ket{\phi_k}$ and $u_k=p_k$. This is not the case for $r> (d+1)/2$.}.

	The situation changes when the rank is increased even further. For example with $r=4$ and $d=5$, the maximal QFI is equal to $p_1(h_1-h_5)^2+\frac{(p_2-p_4)^2}{p_2+p_4}(h_2-h_4)^2$. Further, for a Hilbert space of odd dimension, the vector $\ket{\phi_{(d+1)/2}}=\ket{h_{(d+1)/2}}$ is an eigenstate of the generator: It remains invariant under the dynamics and does not contribute to the QFI. For example for both $r=2$ and $r=3$ with $d=5$, the optimal QFI is given by $p_1(h_1-h_5)^2+p_2(h_2-h_4)^2$.

	We obtained $I_\alpha^*$ by optimizing with respect to unitary
	state preparation while keeping the sensor dynamics
	fixed {(see Fig.~\ref{fig:protocol})}. However, in practice it is often possible not only to 
	manipulate the available state but
	also the sensor dynamics by adding a parameter-independent
	control Hamiltonian
	$H_\mm{c}(t)$ to the original Hamiltonian $H_\param(t)$. 
	While theorem
	\ref{th:main} holds for any 
	{$H_\param(t)$},
	it is an interesting question
	to what extent the maximal QFI {in Eq.~\eqref{eq:opt_QFI_th1}} 
	can be
	increased  
	by adding a
	time-dependent control Hamiltonian.  Again, the
	answer is only known for pure states
	\cite{pang_optimal_2017}. The question, how this generalizes
	if the available state is mixed, brings us to 
	\begin{theorem}\label{th:time_dep}
		For any state $\rho$ with ordered eigenvalues $p_1\geq \dotsm \geq p_d$ and any time-dependent Hamiltonian $H_\alpha(t)$, where $\mu_1(t)\geq\dotsm\geq \mu_d(t)$ are the ordered eigenvalues of $\partial_\alpha H_\alpha(t)\coloneqq\partial H_\alpha(t)/\partial \alpha$, an upper bound for the QFI is given by
		\begin{align}\label{eq:opt_QFI_th2}
		K_\alpha=\frac{1}{2}\sum_{k=1}^{d}p_{k,d-k+1} \left(\int_{0}^{T}[\mu_k(t)-\mu_{d-k+1}(t)]\emph{d}t\right)^2.
		\end{align}
		Let $\ket{\mu_k(t)}$ be the time-dependent eigenvectors of $\partial_\alpha H_\alpha(t)$, $\partial_\alpha H_\alpha(t)\ket{\mu_k(t)}=\mu_k(t)\ket{\mu_k(t)}$. The upper bound $K_\alpha$ is reached by preparing the initial state
		\begin{equation}\label{eq:opt_state_2}
		\rho^*=\sum_{k=1}^d p_k \ket{\phi_k}\bra{\phi_k},
		\end{equation}
		with
		\begin{equation}\label{eq:opt_eigen_th2}
		\ket{\phi_k}=\begin{cases*}
		\frac{\ket{\mu_k(0)}+\ket{\mu_{d-k+1}(0)}}{\sqrt{2}}\quad ~~\text{if } 2k<d+1, \\
		\mathrlap{\ket{\mu_k(0)}}\hphantom{\frac{\ket{\mu_k(0)}+\ket{\mu_{d-k+1}(0)}}{\sqrt{2}}}\quad ~~\text{if }2k=d+1,\\
		\frac{\ket{\mu_k(0)}-\ket{\mu_{d-k+1}(0)}}{\sqrt{2}}\quad ~~\text{if } 2k>d+1,
		\end{cases*}
		\end{equation}
		and choosing the Hamiltonian control $H_\mm{c}(t)$ such that 
		\begin{equation}\label{eq:opt_con_th2}
		U_\alpha(t)\ket{\mu_k(0)}=\ket{\mu_k(t)}\quad \forall k=1,\dotsc,d~~\forall t,
		\end{equation}
		where
		\begin{equation}\label{eq:unitarycon}
		U_\alpha(t)= \mc{T}\left[\exp\left(-\ii\int_0^t [H_\alpha(\tau)+H_\mm{c}(\tau)] \diff \tau\right)\right].
		\end{equation}
	\end{theorem}
	The proof
	(see Supplementary Material \cite{Note2}) 
	starts by rewriting
	$h_\alpha$ as in Ref.~\cite[Eq.~6]{pang_optimal_2017} and	shows that Eq.~\eqref{eq:opt_QFI_th2} is an upper bound
	for Eq.~\eqref{eq:opt_QFI_th1}. We use the Schur convexity
	\cite{marshall1979inequalities} of Eq.~\eqref{eq:opt_QFI_th1} and the
	inequalities from K.~Fan \cite{fan1949theorem,fulton2000eigenvalues} for eigenvalues of the sum of 
	two hermitian matrices.
	
	
	One of the strengths of the bound $K_\alpha$ is that it is given by the eigenvalues of $\partial_\alpha H_\alpha(t)$ and does not depend on the full unitary operator of the sensor dynamics which is hard to calculate for time-dependent Hamiltonians.
	\begin{figure}
		\centering
		\includegraphics[width=0.4\textwidth]{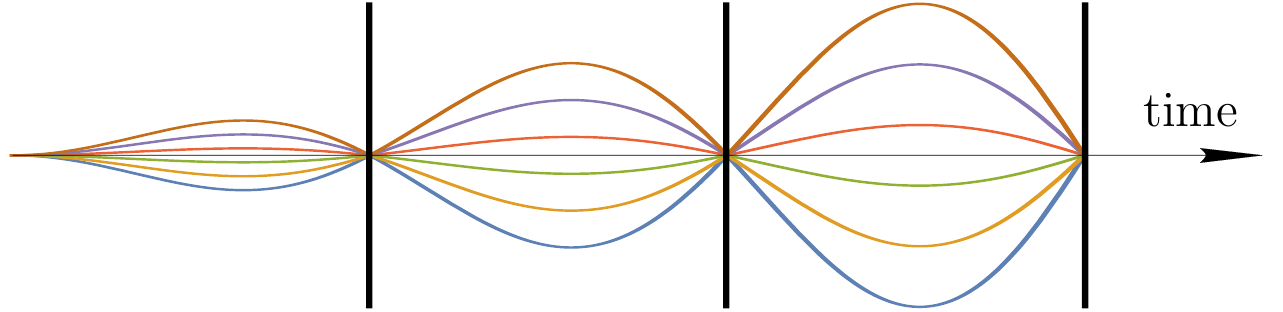}\\
		\caption{Exemplary sketch of time-dependent eigenvalues $\mu_1\geq \dotsm \geq \mu_6$ of
			$\partial H(t)/\partial \omega$ corresponding to $f(t)=\cos(\omega t)$. Vertical black lines indicate the position of single-spin $\pi$-pulses about the $x$-axis in order to interchange eigenvectors $\ket{j,m}\leftrightarrow\ket{j,-m}$.}\label{fig:eigen} 
	\end{figure}
	The optimal initial state with Hamiltonian control in theorem
	\ref{th:time_dep} differs from the optimal initial state without
	Hamiltonian control in theorem \ref{th:main}
	by the fact that the eigenvectors of the
	generator $\dgen$ in Eq.~\eqref{eq:opt_eigen_th1} are replaced by
	those of $\partial_\alpha H_\alpha(0)$ in
	Eq.~\eqref{eq:opt_eigen_th2}. The reason for this is that the optimal
	initial state of theorem 1 is the most sensitive state with respect to
	the sensor dynamics $U_\alpha$. However, if the Hamiltonian is
	time-dependent, the state which is most sensitive to the sensor
	dynamics at time $t$ will also be time-dependent in general. Since the
	Hamiltonian control is allowed to be time-dependent, we can take this
	into account and ensure that the optimal initial state evolves such
	that it is most sensitive to the sensor dynamics for all times
	$t$. This corresponds to the condition in
	Eq.~\eqref{eq:opt_con_th2}. Only in special cases, such as
	phase-shift Hamiltonians $H_\alpha=\alpha G$, we have
	$\dgen=\partial_\alpha H_\alpha$ and, thus, the optimal initial states
	of theorem \ref{th:main} and \ref{th:time_dep} are the same. If they
	are not the same, a
	Hamiltonian $H_\alpha$
	can be seen as suboptimal and
	requires correction by means of the Hamiltonian control in order to
	reach the upper bound of theorem \ref{th:time_dep}.
	
	Formally, the optimal control Hamiltonian from theorem \ref{th:time_dep} depends on the (unknown) parameter $\alpha$. Since we are in the local parameter estimation regime, we have knowledge (from prior measurements) about $\alpha$ such that $\alpha$ can be replaced by the estimate $\hat{\alpha}$. It was shown that replacing $\alpha$ by $\hat{\alpha}$ in the optimal control Hamiltonian does not ruin the benefits from introducing Hamiltonian control \cite{pang_optimal_2017}, and Hamiltonian control was applied experimentally with great success in Ref.~\cite{Schmitt832}.
	For a more detailed discussion of control Hamiltonians we 
	refer to the work of Pang et al.~\cite{pang_optimal_2017}
	\footnote{The optimal control
			Hamiltonian of Pang et
			al.~\cite{pang_optimal_2017} fulfills exactly the condition in
			Eq.~\eqref{eq:opt_con_th2} which makes it optimal not only for pure
			but also for mixed states. For pure states, control
		Hamiltonians that ensure $U_\alpha(t)\ket{\mu_i(0)}=\ket{\mu_i(t)}$
		only for $i=1,d$ would have been sufficient because only $\mu_1$ and
		$\mu_d$ contribute to the upper bound.}.
	\begin{figure*}
		\includegraphics[width=0.312\textwidth]{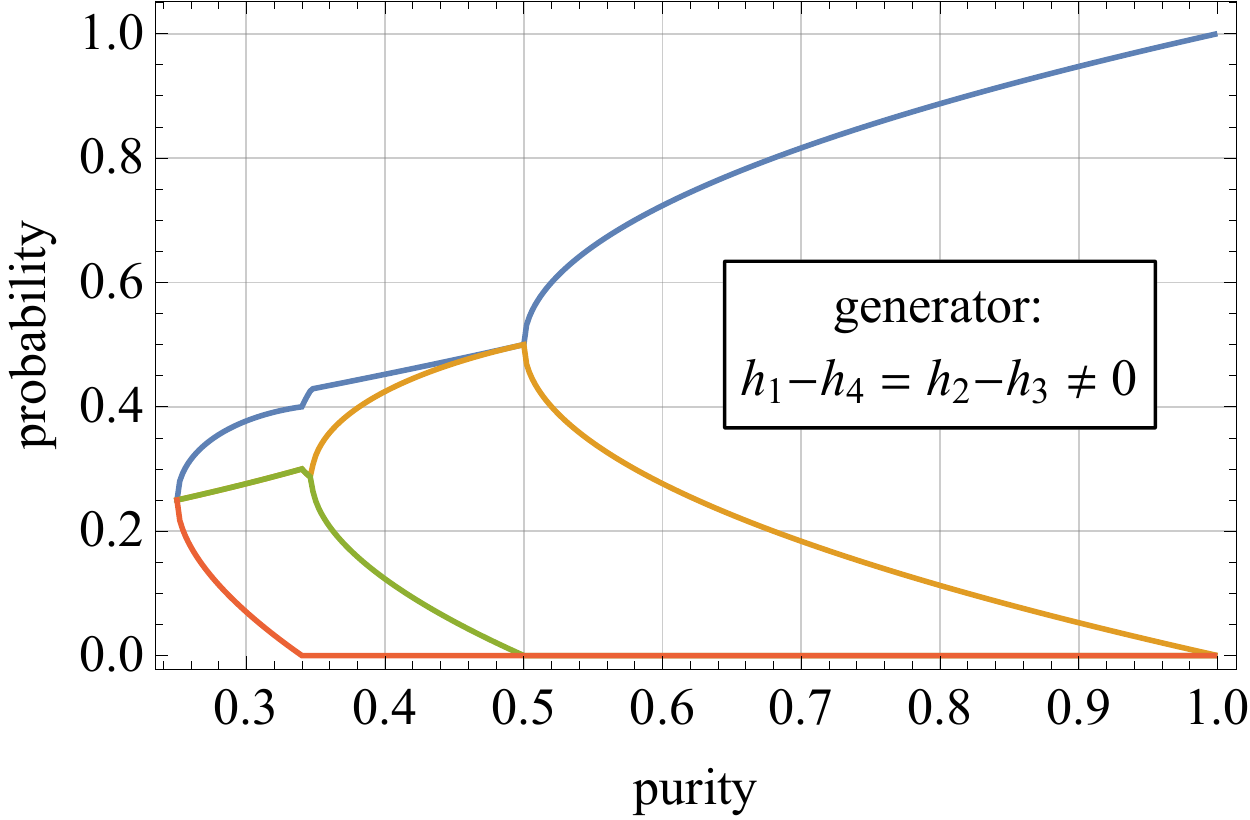}
		\includegraphics[width=0.294\textwidth]{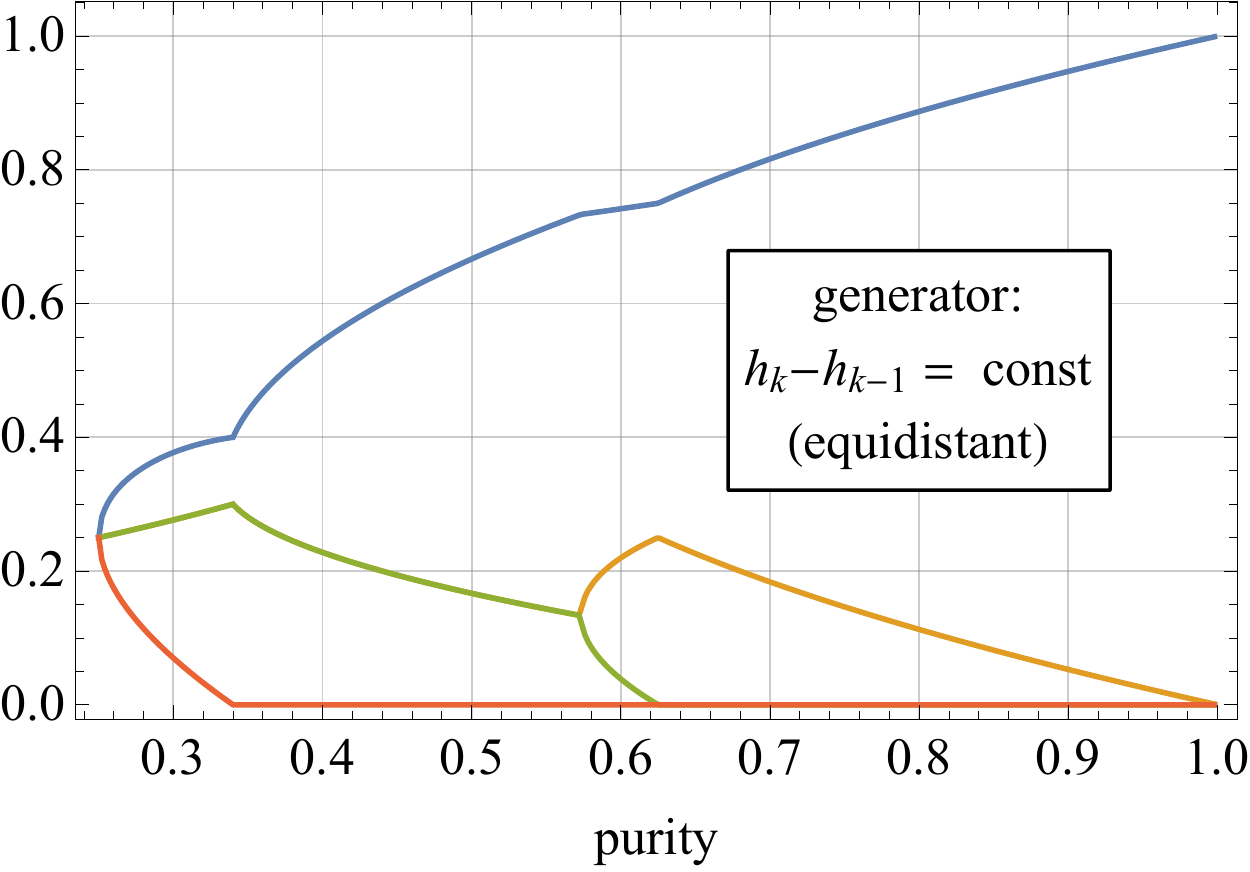}
		\includegraphics[width=0.353\textwidth]{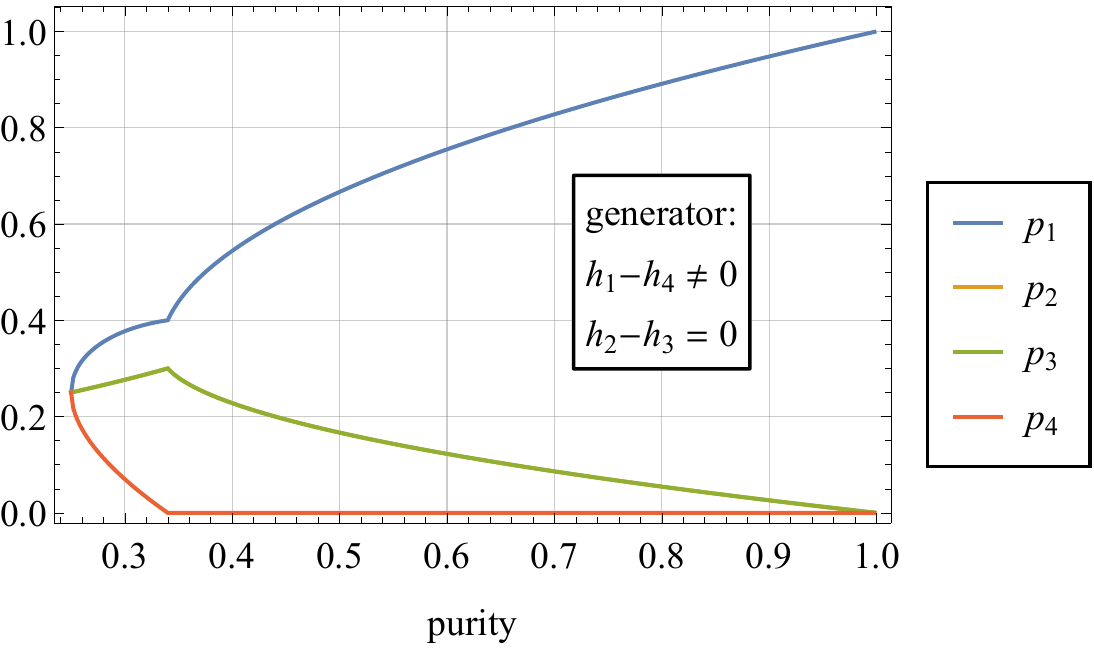}\\
		(a)\hspace{5cm}(b)\hspace{5cm}(c)\hspace{0.5cm}
		\caption{Eigenvalues $p_1\geq \dotsm\geq p_4$ of initial
			two-qubit states that maximize the QFI for different values
			of purity $\gamma$.  For each value of purity, eigenvalues
			$p_i$ are found numerically by maximizing the expression for
			maximal QFI from theorem \ref{th:main} in
			Eq.~\eqref{eq:opt_QFI_th1} under the constraints of fixed
			purity and conservation of probability,
			$\sum_kp_k=1$. Different panels correspond to different
			spectra of the generator with eigenvalues $h_1\geq
			\dotsm\geq h_4$ as indicated in the insets. 
			The generator used in panel (a) has two degeneracies,
			the one in panel (b) has an equidistant spectrum, and
			the one in panel (c) has one degeneracy. In panel (c),
			the line corresponding to $p_3$ overlays the line of $p_2$.}\label{fig:purity}
	\end{figure*}
	
	As applications of our theorems we
consider two
	examples: the  
	estimation of a magnetic field amplitude and the 
	estimation of the frequency of an oscillating magnetic
	field. Both cases  can be described with the general
	Hamiltonian of a system of $N$ spin-$j$ particles subjected to a (time-dependent) magnetic field
	\begin{align}\label{eq:hamilton_n_spin}
	H(t)=\sum_{k=1}^N B f(t)S_z^{(k)}+H_\text{I},
	\end{align}
	with the magnetic field amplitude $B$, some time-dependent real-valued
	modulation function $f(t)$, and spin operator $S_z^{(k)}$ in
	$z$-direction of the $k$th spin. We use the standard angular momentum
	algebra, $S_z^{(k)}\ket{j,m}=m\ket{j,m}$ with
	$m=-j,\dotsc,j$. $H_\text{I}$ is independent of $B$ and takes into
	account possible interactions between spins.  
	This rather general Hamiltonian can be seen as an
	idealization of quantum sensors based on arrays of NV centers \cite{pham2011magnetic, barry2016optical,choi2017quantum}, nuclear
	spin ensembles \cite{de2008nmr}, or vapor of alkali
	atoms \cite{budker2007optical}. 
	Due to imperfect
	polarization and competing depolarization
	effects \cite{scheuer2017robust,appelt1998theory,choi2017depolarization, pagliero_recursive_2014}, the available
	states are mixed.

	Here, we consider the available state of each of the $N$ spins to be
	described by a spin-temperature distribution {(independent of the Hamiltonian in Eq.~\eqref{eq:hamilton_n_spin})}
	\begin{align}
	\rho_\text{th}=\frac{\text{e}^{\beta S_z}}{Z}\,,\label{eq:single_spin}
	\end{align}
	with partition function $Z=\sum_{m=-j}^{j}\text{e}^{\beta m}$, and
	inverse (effective) 
	temperature $\beta$. Eq.~\eqref{eq:single_spin} was derived for {optically polarized} 
	alkali vapors in \cite[Eq.~112]{appelt1998theory},  and we assume that it
	is also a good approximation for the other
	spin-based magnetometers mentioned.
	$\beta$ is related to the degree of polarization $P\in[0,1]$ by $\beta=\ln\frac{1+P}{1-P}$;
	$P=1$ corresponds to a perfectly polarized
	spin in $z$-direction,  described by a pure state, and $P=0$
	corresponds to an unpolarized spin, 
	i.e.~a maximally mixed
	state. 
	The available
	state of the total system is a tensor product of
	spin-temperature distributions,
	$\rho=\rho_\text{th}^{\otimes N}$. 
	
	For the estimation of the amplitude $B$ we assume that the modulation $f(t)$ is known (the case of unknown $f(t)$ would correspond to waveform estimation \cite{tsang2011fundamental, berry2015quantum}). This is naturally the case for (quasi-)constant magnetic fields,
		periodic fields of known frequency, or, for example, when the modulation originates from a relative movement of sensor and environment (the source of $B$) that is tracked separately with another sensor.
	The maximal QFI obtained by using control Hamiltonians (cf.~theorem \ref{th:time_dep}) for estimating the amplitude $B$ is found to be
	\begin{align}
	K_B=g^2(T)\sum_{k=-Nj}^{Nj}q(k)\frac{\sinh^2(\beta k)}{Z^N\cosh(\beta k)}(2k)^2,\label{eq:qfi_opt_spins}
	\end{align}
	where $q(k)$ takes into account the degeneracy of eigenvalues of
	$\rho$ and $\partial_B H(t)\coloneqq\partial H(t)/\partial
	B$. It follows from the definition of the tensor product that the
	degeneracy of the $k$th eigenvalue of both, $\rho$ and $\partial_B
	H(t)$, where eigenvalues are in weakly decreasing order, equals
	the number of possibilities $q(k)$ of getting a sum $k$ when rolling
	$N$ fair dice,
	each having $2j+1$ sides corresponding to values $\{-j,\dotsc, j\}$ (see Supplementary Material \cite{Note2}) \cite[p.~23-24]{uspensky1937introduction}:
	\begin{align}
	q(k)\coloneqq \sum_{\ell=0}^{N} (-1)^{\ell}\binom{N}{\ell}\binom{k+N(j+1)-1-\ell(2j+1)}{N-1},
	\end{align}
	where the binomial coefficient $\binom{a}{b}$ is set to zero if one or both of its coefficients are negative.
	The dependence on measurement time $T$ is given by $g(T)=\int_{0}^{T}|f(t)|
	\text{d}t$.
	
	The QFI in Eq.~\eqref{eq:qfi_opt_spins} exhibits a complicated dependence
	on 
	the number of thermal states $N$ and their spin size $j$. However,
	by deriving a lower bound for Eq.~\eqref{eq:qfi_opt_spins}, we 
	prove
	that the QFI scales $\propto N^2$ for any $j$ as well as $\propto j^2$ for any $N$. {In particular, we find
		$K_B= 4N^2\braket{S_z}^2+\mathcal{O}(N)$ where $\braket{S_z}=\trb{\rho_\text{th} S_z}$, and $\mathcal{O}(N)$  denotes terms of order $N$ and lower order. In the limit of large temperatures, $\braket{S_z}^2$ decays as $\beta^2$ (see Supplementary Material \cite{Note2}).}
	
	This means that Heisenberg scaling {\cite{giovannetti_quantum_2004,holland_interferometric_1993,zwierz_general_2010},} i.e., the quadratic scaling with the system size $j$ or the number of particles $N$, is obtained for the optimal unitary state preparation even if only thermal states are available. Note that this also holds in the context of theorem \ref{th:main} if the generator equals $S_z$. {Importantly,} Heisenberg scaling is found for any finite temperature of the thermal state; only in the limit of infinite temperature, the available state is fully mixed and the QFI vanishes.
	
	In order to attain the QFI \eqref{eq:qfi_opt_spins}, the conditions \eqref{eq:opt_con_th2} must be fulfilled. In particular the Hamiltonian control must cancel interactions between the spins, i.e., $H_\text{I}$
	must be compensated.
	Also, every time the modulation function $f(t)$ changes its sign, we
	must apply a transformation which interchanges the eigenstates
	corresponding to a (degenerate) eigenvalue $e^{\beta k}/Z^N$ of $\rho$
	with the eigenstates corresponding to the (degenerate) eigenvalue
	$e^{-\beta k}/Z^N$ for all $k=1,\dotsc, Nj$. This is realized, for
	instance, with a local $\pi$-pulse about the $x$-axis, which
	interchanges $\ket{j,m}$ and $\ket{j,-m}$ for every single spin.  
	The $\pi$-pulses ensure optimal phase accumulation of the
	optimal state given by Eq.~\eqref{eq:opt_state_2} (cf.~Fig.~\ref{fig:eigen}).
	
	The degeneracy of eigenvalues of $\rho$ and $\partial_B H(t)$
	leads to a freedom in preparing the optimal initial state. 
	The special case of qubits, $j=1/2$, constant magnetic
	field, $f(t)=1$, and no interactions, $H_\text{I}=0$, was studied by
	Modi et al.~\cite{modi2011quantum}. 
	In this case, no
	Hamiltonian control is required which 
	brings us back to
	theorem \ref{th:main}. They conjectured that a unitary state
	preparation consisting of a mixture of GHZ states is optimal in
	their case and calculated the QFI. Theorem 1 confirms their conjecture.

	If, instead of the amplitude, we want to estimate the frequency
	$\omega$ of a periodic magnetic field with known amplitude $B$,
	$f(t)=\cos(\omega t)$, the eigenvalues of
	$\partial H(t)/\partial \omega$ are modulated not with $f(t)$ but with $\partial f(t)/\partial\omega=-t\sin(\omega t)$, see Fig.~\ref{fig:eigen}. The maximal QFI $K_\omega$ equals Eq.~\eqref{eq:qfi_opt_spins} with the only difference that  $g(T)$ is replaced by
	$g_\omega(T)=\int_{0}^{T}Bt|\sin(\omega t)|\text{d}t\simeq BT^2/\pi$,
	corresponding to a $T^4$-scaling of QFI, similar to what was reported in
	Ref.~\cite{pang_optimal_2017}. The optimal control is similar to the
	estimation of $B$: interactions must be canceled and local
	$\pi$-pulses about the $x$-axis must be applied whenever $\partial f(t)/\partial\omega$ crosses zero.
	
	Theorem 1 also allows us to study the problem of optimal initial
	states of given purity $\gamma=\tr \rho^2$. Fixing
	only $\gamma$ amounts 
	to an additional optimization over the spectrum of the initial state, 
	which we solve numerically. As an example, we consider a two-qubit
	system with eigenvalues $p_1\geq \dotsm\geq p_4$, see Fig.~\ref{fig:purity}. 
	We observe that different levels of degeneracy of the spectrum of the generator results in distinct solutions for the optimal eigenvalues $p_k$.

	In conclusion, theorems \ref{th:main} and \ref{th:time_dep} give an
	answer to the question of optimal unitary state preparation  and
	optimal Hamiltonian control for an available mixed state and given
	unitary sensor dynamics that encodes the parameter to be measured in
	the quantum state. In comparison, distilling pure from mixed states at the cost of reducing the number of available probes would be an alternative. However, probes are typically a valuable resource, that is utilized most efficiently along the lines of theorem \ref{th:main} and \ref{th:time_dep}. The two theorems 
	allow one to
	study quantum metrology with mixed states with the same analytical rigor as for
	pure states, and the well-known results about optimal pure states are
	recovered as special cases.
	We find that Heisenberg scaling of the QFI 
	can be reached with thermal states: 
	initial mixedness is not as detrimental as Markovian decoherence during or after the sensor dynamics, which is known to generally destroy the Heisenberg scaling of the QFI \cite{Huelga97,kolodynski2014precision,escher_general_2011}.
	
\begin{acknowledgments}
	L.\,J.\,F.~and D.\,B.~acknowledge
support from the Deutsche Forschungsgemeinschaft (DFG), Grant No.~BR
5221/1-1. J.\,M.\,E.\,F.~was supported by the NRF of Korea, Grant No.~2017R1A2A2A05001422.
\end{acknowledgments}

%

\clearpage
\onecolumngrid
\appendix
\section*{Supplemental Material}
\section{Proof of theorem 1}
\begin{figure}[h!]
	\centering
	\includegraphics[width=0.8\linewidth]{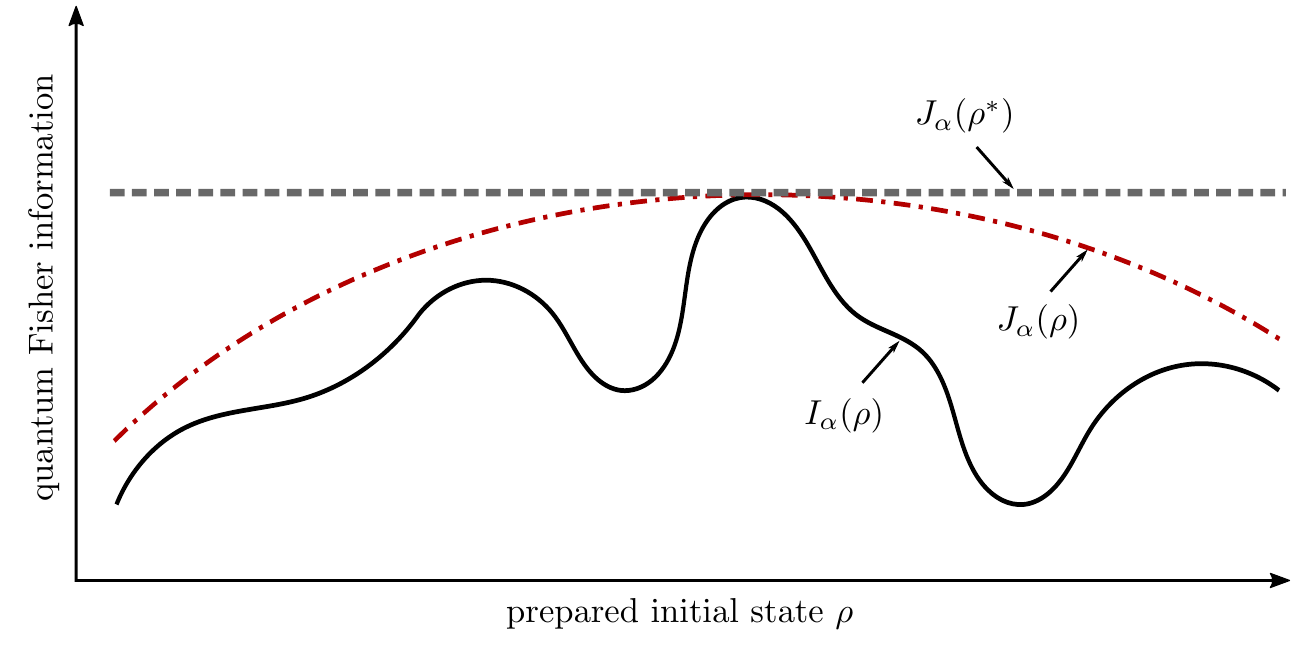}
	\caption{Schematic sketch of the mechanism used to prove theorem 1 from the Letter. First an upper bound $J_\param (\rho)$ (red dash-dotted line) for the quantum Fisher information $I_\param (\rho)$ (black solid line) is constructed. This upper bound is shown to be maximal for $\rho=\rho^*$ (gray dashed line). Then, it is shown that $I_\param (\rho^*)=J_\param (\rho^*)$ from which it follows that $I_\param (\rho^*)$ must be the maximum of $I_\param (\rho)$.}\label{fig:proof_schematic}
\end{figure}
The idea of the proof of theorem 1 from the Letter is the following, see also Fig.~\ref{fig:proof_schematic}:
We carefully construct an upper bound $J_\param (\rho) \geq I_\param (\rho)$ for the QFI $I_\alpha(\rho)$. Then, we show that $(i)$ $J_\param (\rho)$ is maximized by setting $\rho=\rho^*$ and $(ii)$ $J_\param (\rho^*)=I_\param (\rho^*)$. It follows that $I_\param (\rho^*)$ is the maximum of $I_\param (\rho)$.

We first give a technical lemma which introduces inequalities for the $p_{i,j}$ coefficients which are defined as
\begin{align}\label{eq:pij_sup}
\pij{i}{j}\coloneqq\begin{cases*}
0 \text{~\quad\qquad if }p_i=p_j=0,\\
\frac{(p_i-p_j)^2}{p_i+p_j}\text{\quad else.}
\end{cases*}
\end{align}
These inequalities will be used to prove proposition \ref{prop:1} about the existence of coefficients $q_{i,j}$ which fulfill specific conditions. Proposition \ref{prop:1} enables us to find the desired upper bound $J_\param (\rho)$ for the QFI $I_\param (\rho)$. This is then used in the proof of theorem \ref{th:main_sup} which corresponds to theorem 1 from the Letter. 

To facilitate the understanding of the following lemma and proposition, we introduce a schematic arrangement of a set of coefficients $\pij{i}{j}$, see Fig.~\ref{fig:p1}. We consider only coefficients with $1\leq i<j\leq d$ because of the symmetry $\pij{i}{j}=\pij{j}{i}$ and because $\pij{i}{i}=0$.
\begin{figure}[h!]
	\centering
	\includegraphics[scale=1]{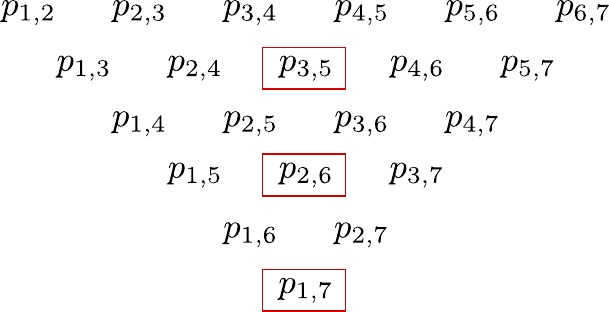}
	\caption{Scheme of $\pij{i}{j}$  for $1\leq i<j\leq d=7$. Coefficients inside the red squared boxes are denoted as \textit{central coefficients}.}\label{fig:p1}
\end{figure}
\begin{lemma}\label{lem:1}
	Let $p_1\geq\dotsm\geq p_d\geq 0$. Then, the following inequalities hold:
	\begin{align*}
	(i)& \qquad\qquad \, p_{i,l}\geq p_{i,j}+p_{j,l} ~~\, \quad \text{for }1\leq i<j<l\leq d,\\
	(ii)& \quad p_{i,l}-p_{i+1,l}\geq p_{i,k}-p_{i+1,k} ~~~ \text{for }1<i+1<k<l\leq d,\\
	(iii)& \quad p_{i,l}-p_{i,l-1}\geq p_{j,l}-p_{j,l-1} ~~~\,\,  \text{for }1\leq i<j<l-1<d.
	\end{align*}
\end{lemma}
\begin{proof}
	First we prove that
	\begin{equation}\label{eq:main_ineq}
	\pij{i}{l}+\pij{j}{k}\geq \pij{i}{k}+\pij{j}{l}\quad \text{for }1\leq i<j<k<l\leq d.
	\end{equation}
	If $p_i\geq p_j=p_k=p_l=0$, inequality \eqref{eq:main_ineq} holds trivially. Otherwise, we find
	\begin{align}
	\pij{i}{l}+\pij{j}{k}- \pij{i}{k}-\pij{j}{l}
	=\frac{4(p_i-p_{j})(p_{k}-p_{l})(p_{i}p_{j}p_{k}+p_{j}p_{k}p_{l}+p_{i}(p_{j}+p_{k})p_{l})}{(p_i+p_{k})(p_{j}+p_{k})(p_i+p_{l})(p_{j}+p_{l})},
	\end{align}
	which is clearly nonnegative because all factors in the denominator are positive and all factors in the numerator  are nonnegative. This proves inequality \eqref{eq:main_ineq}.
	
	Inequalities $(i)$, $(ii)$, and $(iii)$ from the lemma are special cases of inequality \eqref{eq:main_ineq}:	
	If $p_j=p_k$, inequality \eqref{eq:main_ineq} holds also for $j=k$ and it follows inequality $(i)$. Further, from inequality \eqref{eq:main_ineq} we find
	$\pij{i}{l}-\pij{j}{l} \geq \pij{i}{k}-\pij{j}{k}$ which for $j=i+1$ gives inequality $(ii)$, and we find $\pij{i}{l}-\pij{i}{k} \geq \pij{j}{l}-\pij{j}{k}$ which for $k=l-1$ gives inequality $(iii)$.
\end{proof}

In analogy to the coefficients $p_{i,j}$, we introduce another set of coefficients defined by
\begin{equation}\label{eq:qij}
\qij{i}{j}\coloneqq \sum_{k=i}^{j-1} q_{k,k+1},
\end{equation}
for $i<j$. This means that the set of coefficients $\ens{\qij{i}{j}}$ is fully defined by the coefficients $\qij{i}{i+1}$ with $i=1,\dotsc, d-1$.

\begin{proposition}\label{prop:1}
	For any dimension $d\geq 2$ and for any $p_1\geq \dotsm \geq p_d\geq 0$, there exist coefficients $q_{k,k+1}\geq 0$ with $1\leq k\leq d-1$ such that for $1\leq i<j\leq d$:
	\begin{align}
	q_{i,j}&=p_{i,j}\quad \text{if~~} j=d-i+1\qquad\text{(central coefficients),}\label{eq:con1}\\
	q_{i,j}&\geq p_{i,j}\quad\text{else,}\label{eq:con2}
	\end{align}
	where coefficients $p_{i,j}$ and $q_{i,j}$ are defined in Eqs.\,\eqref{eq:pij_sup} and \eqref{eq:qij}, respectively.
\end{proposition}

\begin{proof}
	The  proof works by induction in dimension $d$, once for even $d$ and once for odd $d$.\\
	\underline{Even dimension $d$}
	\paragraph*{Base case $d=2$:}
	There is only one coefficient $\pij{i}{j}$ with $1\leq i<j\leq 2$, which is $\pij{1}{2}$. The proposition for $d=2$ holds because $\qij{1}{2}\coloneqq\pij{1}{2}$ fulfills conditions \eqref{eq:con1} and \eqref{eq:con2} trivially.
	\paragraph*{Inductive step:}
	Suppose the proposition holds for $d=n$. We will prove the proposition for $d=n+2$. 
	
	First, the induction hypothesis is applied to $n$ coefficients $p_2,\dotsc, p_{n+1}$: For any $p_2\geq \dotsm \geq p_{n+1}\geq 0$, there exist coefficients $q_{k,k+1}\geq 0$ for $2\leq k\leq n$ such that for $2\leq i<j\leq n+1$:
	\begin{align}
	q_{i,j}&=p_{i,j}\quad \text{if~~} j=n-i+3\qquad\text{(central coefficients),}\label{eq:hypo1}\\
	q_{i,j}&\geq p_{i,j}\quad\text{else.}\label{eq:hypo2}
	\end{align}
	
	Second, we show that for any $p_1$ and $p_{n+2}$ with $p_1\geq p_2$ and $p_{n+1}\geq p_{n+2}\geq 0$ there exist two further coefficients $q_{1,2}$ and $q_{n+1,n+2}$ such that 
	\begin{align}
	q_{1,n+2}&=p_{1,n+2}\quad \text{(central coefficients),}\label{eq:sub1}\\
	q_{1,j}&\geq p_{1,j}\quad \text{for }j=2,\dotsc,n+1\quad  \text{(left flank),}\label{eq:sub2}\\
	q_{i,n+2}&\geq p_{i,n+2}\quad \text{for }i=2,\dotsc,n+1\quad  \text{(right flank).}\label{eq:sub3}
	\end{align}
	A graphical visualization of the inductive step is shown in Fig.~\ref{fig:p3} which explains the terms left flank and right flank used to designate the inequalities above.
	\begin{figure}[h!]
		\centering
		\includegraphics[scale=1]{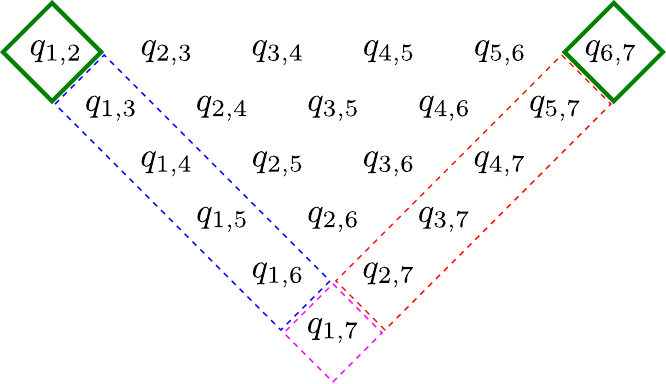}\hspace{1cm}\includegraphics[scale=1]{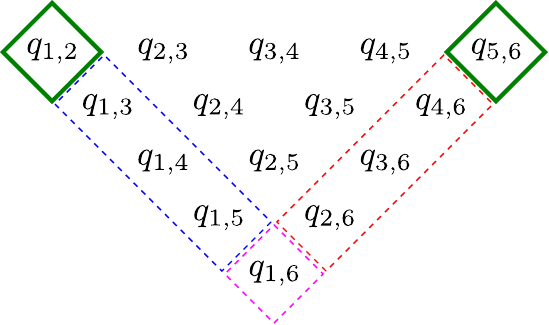}
		\caption{Recursion steps from $d=5$ to $d=7$ (left) and from $d=4$ to $d=6$ (right). In the green squares are the two new elements we need to choose. In the blue (resp.~red) rectangles are the new left (resp.~right) flanks that need to fulfill conditions \eqref{eq:con2} ; in the magenta squares are the new central coefficients that need to fulfill condition \eqref{eq:con1}.}\label{fig:p3}
	\end{figure}
	The existence of $q_{1,2}$ and $q_{n+1,n+2}$ such that conditions \eqref{eq:sub1},\eqref{eq:sub2}, and \eqref{eq:sub3} hold is shown explicitly by setting
	\begin{align}
	\qij{1}{2}&\coloneqq \pij{1}{n+2}-\pij{2}{n+2}\label{eq:setq1},\\
	\qij{n+1}{n+2}&\coloneqq \pij{2}{n+2}-\pij{2}{n+1}\label{eq:setq2},
	\end{align}
	and checking conditions \eqref{eq:sub1},\eqref{eq:sub2}, and \eqref{eq:sub3}:
	We find
	\begin{align*}
	\qij{1}{n+2}&= q_{1,2}+q_{2,n+1}+q_{n+1,n+2}&\text{[Eq.~\eqref{eq:qij}]}\\
	&=p_{1,n+2}-p_{2,n+2}+p_{2,n+1}+p_{2,n+2}-p_{2,n+1}&\text{[Eqs.\,\eqref{eq:setq1},\eqref{eq:setq2}, and Eq.~\eqref{eq:hypo1} for $i=2$]}\\
	&=p_{1,n+2}&
	\end{align*}
	which fulfills the condition for central coefficients  [condition \eqref{eq:sub1}]. Further, for $j=3,\dotsc,n+1$:
	\begin{align*}
	q_{1,j}&=q_{1,2}+q_{2,j}&\text{[Eq.~\eqref{eq:qij}]}\\
	&\geq \pij{1}{n+2}-\pij{2}{n+2}+p_{2,j}&\text{[Eq.~\eqref{eq:setq1} and inequality \eqref{eq:hypo2}]}\\
	&\geq p_{1,j}-p_{2,j}+p_{2,j}&\text{[lemma \ref{lem:1} $(ii)$]}\\
	&=p_{1,j},
	\end{align*}
	and
	\begin{align*}
	q_{1,2}&=\pij{1}{n+2}-\pij{2}{n+2}&\text{[Eq.~\eqref{eq:setq1}]}\\
	&\geq p_{1,2}+p_{2,n+2}-\pij{2}{n+2}&\text{[lemma \ref{lem:1} $(i)$]}\\
	&=p_{1,2},&
	\end{align*}
	which fulfill the conditions for the left flank [condition \eqref{eq:sub2}].
	The proof for the right flank is similar:  For $i=2,\dotsc,n$:
	\begin{align*}
	q_{i,n+2}&=q_{i,n+1}+q_{n+1,n+2}&\text{[Eq.~\eqref{eq:qij}]}\\
	&\geq p_{i,n+1}+\pij{2}{n+2}-\pij{2}{n+1}&\text{[Eq.~\eqref{eq:setq2} and inequality \eqref{eq:hypo2}]}\\
	&\geq p_{i,n+1}+\pij{i}{n+2}-\pij{i}{n+1}&\text{[lemma \ref{lem:1} $(iii)$]}\\
	&=\pij{i}{n+2},
	\end{align*}
	and
	\begin{align*}
	q_{n+1,n+2}&=\pij{2}{n+2}-\pij{2}{n+1}&\text{[Eq.~\eqref{eq:setq2}]}\\
	&\geq \pij{2}{n+1}+\pij{n+1}{n+2}-\pij{2}{n+1}&\text{[lemma \ref{lem:1} $(i)$]}\\
	&=\pij{n+1}{n+2},&
	\end{align*}
	which fulfill the conditions for the right flank [condition \eqref{eq:sub3}]. This proves the proposition  for $d=n+2$, concluding the proof by induction for even dimensions.\\
	\noindent\underline{Odd dimension $d$}
	\paragraph*{Base case d=3:}
	There are only three coefficients $\pij{i}{j}$ with $1\leq i<j\leq 3$, which are $\pij{1}{2},\pij{2}{3}$, and $\pij{1}{3}$. The proposition for $d=3$ holds because  $\qij{1}{2} \coloneqq \pij{1}{3} -\pij{2}{3}$, $\qij{2}{3} \coloneqq \pij{2}{3}$, and $\qij{1}{3}=\qij{1}{2}+\qij{2}{3}$ fulfill the conditions \eqref{eq:con1} and \eqref{eq:con2}: $\qij{1}{2} = \pij{1}{3} -\pij{2}{3}\geq\pij{1}{2} +\pij{2}{3}-\pij{2}{3}=\pij{1}{2}$ where inequality $(i)$ from lemma \ref{lem:1} was used, while the other conditions hold trivially. 
	\paragraph*{Inductive step:} Analog to the inductive step for even $d$.
\end{proof}

Equipped with proposition \ref{prop:1} we can prove theorem 1 from the Letter:
\begin{theorem}\label{th:main_sup}
	For any state $\rho$ and any generator $\dgen$ with ordered eigenvalues $p_1\geq \dotsm \geq p_d$ and $h_1\geq\dotsm\geq h_d$, respectively, the maximal QFI with respect to all unitary state preparations $U\rho U^\dagger$, $U\in\emph{U}(d)$, is given by
	\begin{align}\label{eq:opt_QFI_th1_sup}
	I_\alpha^*\coloneqq\max_{U} \qfinkd_\param (U\rho U^\dagger)=\frac{1}{2}\sum_{k=1}^{d}\pij{k}{d-k+1}(h_k-h_{d-k+1})^2.
	\end{align}
	Let $\ket{h_k}$ be the eigenvectors of the generator,  $\dgen\ket{h_k}=h_k\ket{h_k}$.
	The maximum $I_\alpha^*$ is obtained by preparing the initial state
	\begin{equation}\label{eq:opt_state_sup}
	\rho^*\coloneqq\sum_{k=1}^d p_k \ket{\phi_k}\bra{\phi_k}
	\end{equation}
	with
	\begin{equation}\label{eq:opt_eigen_th1_sup}
	\ket{\phi_k}\coloneqq\begin{cases*}
	\frac{\ket{h_k}+e^{\ii\chi_k}\ket{h_{d-k+1}}}{\sqrt{2}}\quad ~~\text{if } 2k<d+1, \\
	\mathrlap{\ket{h_k}}\hphantom{\frac{\ket{h_k}+\ket{h_{d-k+1}}}{\sqrt{2}}}\quad \quad~~~\text{if }2k=d+1,\\
	\frac{\ket{h_k}-e^{\ii\chi_k}\ket{h_{d-k+1}}}{\sqrt{2}}\quad ~~\text{if } 2k>d+1.
	\end{cases*}
	\end{equation}
	where $\chi_k$ are arbitrary real phases (the theorem as formulated in the Letter is recovered by setting $\chi_k=0$).
\end{theorem}
\begin{proof}
	First we reformulate the optimization problem in a more convenient way:
	
	The unitary state preparation $U\rho U^\dagger$ has invariant eigenvalues for all $U\in\text{U}(d)$. However, the unitary freedom $U\in\text{U}(d)$ allows one to change the basis from the ordered orthonormal basis of eigenvectors $(\ket{\psi_i})_{i=1}^d$ of $\rho$, where $\rho\ket{\psi_i}=p_i\ket{\psi_i}$, to any other ordered orthonormal basis. Therefore, the optimization problem with respect to unitary transformations $U\in\text{U}(d)$ on the state $\rho$ is equivalent to optimizing over ordered bases $B\in S$ where
	\begin{align}
	S\coloneqq\{(\ket{\xi_i})_{i=1}^d : \braket{\xi_i|\xi_j}=\delta_{i,j}~\forall i,j\in\{1,\dotsc,d\}\}.
	\end{align}
	Note that the ordering of eigenvectors corresponds to the ordering of eigenvalues $p_i$ which plays a crucial role in the theorem. 
	The basis corresponding to $\rho^*$ is given by $B^*=(\ket{\phi_i})_{i=1}^d$, and
	the maximization in Eq.~\eqref{eq:opt_QFI_th1_sup} is equivalent to
	\begin{align}\label{eq:max_qfi_reformulated}
	I_\alpha^*\coloneqq\max_{ B \in S}I_\param ( B ),
	\end{align}
	where the QFI was redefined as a function of $B$:
	\begin{align}
	I_\alpha( B)&=2\sum_{i,j=1}^d \pij{i}{j}|\left[h_{\alpha} ( B )\right]_{i,j}|^2.
	\end{align}
	The coefficients $p_{i,j}$ are defined in Eq.~\eqref{eq:pij_sup} with respect to the eigenvalues $p_i$ and $\left[h_{\alpha} ( B )\right]_{i,j}=\bra{\xi_i} h_{\alpha} \ket{\xi_j}$ are the coefficients of $h_{\alpha}$ with respect to $B=(\ket{\xi_i})_{i=1}^d$.
	
	In order to prove that the maximum is reached  by $B^*$, we introduce an upper bound for the QFI. We start by rewriting the QFI, exploiting the symmetries $p_{i,j}=p_{j,i}$ and $\absv{\left[h_{\alpha} ( B )\right]_{i,j}}^2=\absv{\left[h_{\alpha} ( B )\right]_{j,i}}^2$:
	\begin{align}
	\qfinkd_\param (B)&=2\sum_{i,j=1}^{d} \pij{i}{j}\absv{\left[h_{\alpha} ( B )\right]_{i,j}}^2\\
	&=4\sum_{i=1}^{d-1} \sum_{j=i+1}^{d}\pij{i}{j} \absv{\left[h_{\alpha} ( B )\right]_{i,j}}^2\label{eq:replace}.
	\end{align}
	Then, an upper bound for $I_\alpha(B)$ is obtained by replacing coefficients $p_{i,j}$ in Eq.~\eqref{eq:replace} with new coefficients $q_{i,j}\geq p_{i,j}$ for all $1\leq i<j\leq d$:
	\begin{align}
	\qfinkd_\param (B)\leq J_\param (B)\coloneqq 4\sum_{i=1}^{d-1} \sum_{j=i+1}^{d}\qij{i}{j} \absv{\left[h_{\alpha} ( B )\right]_{i,j}}^2,\label{eq:J} 
	\end{align}
	where  $J_\param (B)$ denotes the upper bound. We choose coefficients $q_{i,j}$ according to proposition \ref{prop:1}, i.e., besides $q_{i,j}\geq p_{i,j}$ they fulfill $p_{i,j}=q_{i,j}$ for $j=d-i+1$ and $q_{i,j}=\sum_{k=i}^{j-1}q_{k,k+1}$ for all $1\leq i<j\leq d$. We rewrite the upper bound $J_\param (B)$:
	\begin{align}
	J_\param (B)&=4\sum_{i=1}^{d-1} \sum_{j=i+1}^{d}\qij{i}{j} \absv{\left[h_{\alpha} ( B )\right]_{i,j}}^2\\
	&=4\sum_{i=1}^{d-1} \sum_{j=i+1}^{d}\sum_{k=i}^{j-1} q_{k,k+1} \absv{\left[h_{\alpha} ( B )\right]_{i,j}}^2\\
	&=4\sum_{k=1}^{d-1}q_{k,k+1}\sum_{i=1}^{k} \sum_{j=k+1}^{d}  \absv{\left[h_{\alpha} ( B )\right]_{i,j}}^2\\
	&=4\sum_{k=1}^{d-1}q_{k,k+1} \pnorm{h_\param( B ,k)}{2}^2,\label{eq:last_J}
	\end{align}
	where $h_\param( B ,k)$ denotes the subblock of $h_\param( B )$ with coefficients from the 1st to the $k$th row and from the $(k+1)$th to the $d$th column, and  $\pnorm{\cdot}{2}^2$ denotes the Hilbert--Schmidt norm which is defined for a $m\times n$ matrix $A$ as $\pnorm{A}{2}^2\coloneqq\trb{A^\dagger A}= \sum_{i,j=1}^{m,n} \absv{A_{i,j}}^2$. Since $h_\param( B )$ is Hermitian it divides in subblocks as
	\begin{align}
	h_\param( B )=\begin{pmatrix}
	\bullet &h_\param( B ,k)\\
	h_\param^\dagger ( B ,k)& \bullet\\
	\end{pmatrix},
	\end{align}
	where the quadratic subblocks on the diagonal are not further specified.

	Next, we maximize the upper bound $J_\param(B)$ and show that it equals the QFI at its maximum.
	In order to maximize $J_\param(B)$, we use the Bloomfield--Watson inequality \cite{bloomfield1975inefficiency} on the Hilbert--Schmidt norm of off-diagonal blocks such as $h_\param ( B ,k)$. We take a convenient formulation of the inequality from Ref.\cite[Eqs.\,(1.14) and (4.3)]{drury2002some} and apply it to $h_\param( B ,k)$:
	\begin{equation}\label{eq:Bloom}
	\pnorm{h_\param( B ,k)}{2}^2 \leq \frac{1}{4} \sum_{i=1}^{m(k)} (h_i-h_{d-i+1})^2,
	\end{equation}
	where $m(k)=\min(k,d-k)$. We evaluate the left-hand side of the Bloomfield--Watson inequality \eqref{eq:Bloom} for $B=B^*$, where $B^*$ is the eigenbasis of $\rho^*$, defined above Eq.~\eqref{eq:max_qfi_reformulated}:
	\begin{align}
	\pnorm{h_\param( B ^*,k)}{2}^2 &=\sum_{i=1}^{k} \sum_{j=k+1}^{d}  \absv{\left[h_{\alpha} ( B ^*)\right]_{i,j}}^2\\
	&=\sum_{i=1}^{k} \sum_{j=k+1}^{d}  \absv{\bra{\phi_i} \dgen \ket{\phi_j}}^2 \label{eq:step1}\\
	&= \sum_{i=1}^{k} \sum_{j=k+1}^{d} \left(\delta_{i,j}\frac{h_i+h_{d-i+1}}{2} + \delta_{i,d-j+1}\frac{|h_i-h_{d-i+1}|}{2}\right)^2 \label{eq:step2}\\
	&= \frac{1}{4}\sum_{i=1}^{m(k)} (h_i-h_{d-i+1})^2, \label{eq:step3}
	\end{align}
	where we used the definition of $\ket{\phi_i}$ [Eq.~\eqref{eq:opt_eigen_th1_sup}] to get from Eq.~\eqref{eq:step1} to \eqref{eq:step2}.
	In Eq.~\eqref{eq:step2}, the first summand (within the brackets) evaluates always to zero while the second summand is nonzero in $m(k)$ cases as given in Eq.~\eqref{eq:step3}. 
	Note, that Eq.~\eqref{eq:step3} equals the right-hand side of inequality \eqref{eq:Bloom}. Therefore, the Bloomfield--Watson inequality \eqref{eq:Bloom} is saturated for $B=B^*$ and, in particular, $\pnorm{h_\param( B,k)}{2}^2\leq \pnorm{h_\param( B ^*,k)}{2}^2$ for all $B\in S$.
	
	This implies $J_\param (B)\leq J_\param (B^*)$  for all $B\in S$ which can be seen from Eq.~\eqref{eq:last_J} and by realizing that the coefficients
	$\qij{i}{j}$ in $J_\param (B)$ are nonnegative, which follows from the nonnegativity of $\pij{i}{j}$. Thus, $J_\param (B^*)$ is the maximum of  $J_\param (B)$ with respect to $B$.
	
	Now, we show that $J_\param (B^*)=I_\param (B^*)$ starting from the definition of $J_\param (B)$ in Eq.~\eqref{eq:J}:
	\begin{align}
	J_\param (B ^*)&=4\sum_{i=1}^{d-1} \sum_{j=i+1}^{d}\qij{i}{j} \left(\delta_{i,j}\frac{h_i+h_{d-i+1}}{2} + \delta_{i,d-j+1}\frac{h_i-h_{d-i+1}}{2}\right)^2\\
	&=\sum_{i=1}^{d-1} \sum_{j=i+1}^{d}\qij{i}{j}\delta_{i,d-j+1}(h_i-h_{d-i+1})^2  \label{eq:step4}\\
	&=\frac{1}{2}\sum_{i=1}^{d}\pij{i}{i-d+1}(h_i-h_{d-i+1})^2=I_\param (B ^*),\label{eq:onehalf}
	\end{align}
	where we used $\qij{i}{i-d+1}=\pij{i}{i-d+1}$ and, to get from Eq.~\eqref{eq:step4} to \eqref{eq:onehalf}, we first came back to a summation over all $1\leq i,j\leq d$ before evaluating $\delta_{i,d-j+1}$ which explains the factor $1/2$ in Eq.~\eqref{eq:onehalf}.
	
	It follows from $J_\param (B )\geq I_\param (B )~\forall B\in S$ that $\displaystyle \max_{ B \in S}J_\param (B )\geq \max_{ B \in S}I_\param (B )$, and, then, it follows from $\displaystyle\max_{ B \in S}J_\param (B )=I_\param (B ^*)$ that $I_\param (B ^*)$ is the maximum of $I_\param (B)$ with respect to $B\in S$.
\end{proof}

\section{Proof of theorem 2}
Let us first introduce some notation. The real, nonnegative coordinate space of $d$ dimensions is denoted by $\mathbb{R}_+^d$. 
For two vectors $\boldsymbol{x},\boldsymbol{y}\in \mathbb{R}_+^d$, the element-wise vector ordering $x_i\leq y_i$ for all $i\in\{1,\dotsc, d\}$ is denoted as $\boldsymbol{x}\leq \boldsymbol{y}$. 
For any $\boldsymbol{x}\in \mathbb{R}^d_+$, let $x_{[1]},\dotsc,x_{[d]}$ be the components of $\boldsymbol{x}$ in decreasing order, and let
\begin{align}
\boldsymbol{x}_\downarrow \coloneqq(x_{[1]},\dotsc,x_{[d]})
\end{align}
denote the decreasing rearrangement of $\boldsymbol{x}$.
Let
\begin{align}
\mathcal{D}_+^d&\coloneqq\{(x_1,\dotsc,x_d):x_1\geq \dotsm\geq x_d\geq 0\}
\end{align}
be the set of decreasing rearrangements of elements from $ \mathbb{R}^d_+$.
\begin{definition}
	For a hermitian matrix $X$ with eigenvalues $x_1\geq x_2\geq\dotsm\geq x_d$ define 
	\begin{align}
	\boldsymbol{d}(X)\coloneqq(x_1-x_d,x_2-x_{d-1}\dotsc,x_{\ceil*{d/2}}-x_{d-\ceil*{d/2}+1}),
	\end{align}
	where $\ceil*{d/2}$ denotes the smallest integer $j$ with $j\geq d/2$.
\end{definition}
Note that the entries of $\boldsymbol{d}(X)$ are nonnegative and in decreasing order, i.e., $\boldsymbol{d}(X)\in \mathcal{D}_+^{\ceil*{d/2}}$.
\begin{definition}
	Let $\boldsymbol{x},\boldsymbol{y}\in\mathbb{R}_+^d$. We say that $\boldsymbol{x}$ is weakly majorized by $\boldsymbol{y}$, denoted by $\boldsymbol{x}\prec_w\boldsymbol{y}$, if 
	\begin{align}
	\sum_{i=1}^kx_{[i]}\leq \sum_{i=1}^ky_{[i]}\quad \forall k=1,\dotsc, d.
	\end{align}
\end{definition}
\begin{lemma}\label{lem:weakmajo}
	Let $A$, $B$, and $C=A+B$ be hermitian matrices with eigenvalues $a_1\geq\dotsm \geq a_d$, $b_1\geq\dotsm \geq b_d$, and $c_1\geq\dotsm\geq c_d$, respectively.  Then,	$\boldsymbol{d}(C)\prec_w\boldsymbol{d}(A)+\boldsymbol{d}(B)$. 
\end{lemma}
\begin{proof}
	The inequalities of K.~Fan (see for instance \cite[eq.3]{fulton2000eigenvalues}) for the eigenvalues of $A$, $B$, and $C=A+B$ are
	\begin{align}\label{eq:ineq1}
	\sum_{i=1}^{r}c_i\leq \sum_{i=1}^{r}a_i+b_i\quad \forall r=1,\dotsc,d-1.
	\end{align}
	Subtracting them from the trace condition
	\begin{align}
	\sum_{i=1}^{d}c_i= \sum_{i=1}^{d}a_i+b_i
	\end{align}
	and rearranging the indices gives
	\begin{align}\label{eq:ineq2}
	\sum_{i=1}^{r}c_{d-i+1}\geq \sum_{i=1}^{r}a_{d-i+1}+b_{d-i+1}\quad \forall r=1,\dotsc,d-1.
	\end{align}
	Subtracting inequality \eqref{eq:ineq2} from inequality \eqref{eq:ineq1} gives
	\begin{align}
	\sum_{i=1}^{r}c_i-c_{d-i+1}\leq \sum_{i=1}^{r}a_i-a_{d-i+1}+b_i-b_{d-i+1}\quad \forall r=1,\dotsc,d-1,
	\end{align}
	which are for $r=1,\dotsc,\ceil*{d/2}$ the weak majorization conditions for $\boldsymbol{d}(C)\prec_w\boldsymbol{d}(A)+\boldsymbol{d}(B)$.
\end{proof}
\begin{definition}\label{def:f}
	For any $\boldsymbol{p}\in\mathcal{D}_+^d$  define 
	\begin{align}
	\phi_{\boldsymbol{p}}: \mathbb{R}_+^d \rightarrow \mathbb{R}, \phi_{\boldsymbol{p}}(\boldsymbol{x})\coloneqq\sum_{i=1}^d p_ix_{[i]}^2.
	\end{align}
\end{definition}
\begin{lemma}\label{lem:schur}
	For any $\boldsymbol{p}\in\mathcal{D}_+^d$, 
	$\phi_{\boldsymbol{p}}$ is increasing and Schur convex on $\mathbb{R}_+^d$, i.e., the following conditions hold \emph{\cite[part I,ch.3,A.4]{marshall1979inequalities}}:
	\begin{enumerate}[(i)]
		\item $\boldsymbol{x}\leq \boldsymbol{y}\Rightarrow \phi_{\boldsymbol{p}}(\boldsymbol{x})\leq \phi_{\boldsymbol{p}}(\boldsymbol{y})$ (increasing),
		\item $\phi_{\boldsymbol{p}}(\boldsymbol{x})$ is invariant under permutation of coefficients of $\boldsymbol{x}$ for any $\boldsymbol{x}\in\mathbb{R}^d_+$ (symmetric),
		\item $(x_i-x_j)\left(\frac{\partial \phi_{\boldsymbol{p}}(\boldsymbol{x})}{\partial x_i}-\frac{\partial \phi_{\boldsymbol{p}}(\boldsymbol{x})}{\partial x_j}\right)\geq 0$ $\forall \boldsymbol{x}\in \mathbb{R}_+^d$ and $\forall i\neq j$ (Schur's condition).
	\end{enumerate}
\end{lemma}
\begin{proof}
	From $\boldsymbol{x}\leq \boldsymbol{y}$ it follows that $x_{[i]}\leq y_{[i]}~\forall i$, which implies $p_ix_{[i]}^2\leq p_iy_{[i]}^2~\forall i$ for any $p_i\geq 0$. Finally it follows
	$\sum_i p_ix_{[i]}^2\leq \sum_ip_iy_{[i]}^2$ which proves condition $(i)$.
	Condition $(ii)$ follows directly from the definition of $\phi_{\boldsymbol{p}}$.
	Finally, we have
	\begin{align}
	(x_i-x_j)\left(\frac{\partial \phi_{\boldsymbol{p}}(\boldsymbol{x})}{\partial x_i}-\frac{\partial \phi_{\boldsymbol{p}}(\boldsymbol{x})}{\partial x_j}\right)=(x_i-x_j)2(qx_i-rx_j),
	\end{align}
	where $q,r$ are some components of $\boldsymbol{p}$ with 
	$q\geq r$ if $x_i\geq x_j$ and $q\leq r$ if $x_i\leq x_j$ due to the definition of $\phi_{\boldsymbol{p}}$.
	It follows condition $(iii)$.	
\end{proof}
\begin{lemma}\label{lem:Schurconvex}
	Let $A$, $B$, and $C=A+B$ be Hermitian matrices. For any $\boldsymbol{p}\in\mathcal{D}_+^d$,
	\begin{align}
	\boldsymbol{d}(C)\prec_w\boldsymbol{d}(A)+\boldsymbol{d}(B)~\Rightarrow~\phi_{\boldsymbol{p}}(\boldsymbol{d}(C))\leq \phi_{\boldsymbol{p}}(\boldsymbol{d}(A)+\boldsymbol{d}(B)).
	\end{align}
\end{lemma}
\begin{proof}
	The proof follows from a theorem given in Ref.~\cite[part I,ch.3,A.8]{marshall1979inequalities} about weak majorization and lemma \ref{lem:schur}.
\end{proof}
We are now ready to prove the following inequality:
\begin{lemma}\label{cor:triangle}
	Let $\boldsymbol{p}\in \mathcal{D}_+^d$, and let $p_{i,j}$ be defined as in Eq.~\eqref{eq:pij_sup} for the components of $\boldsymbol{p}$.
	Let $A$, $B$, and $C=A+B$ be Hermitian matrices with eigenvalues $a_1\geq\dotsm \geq a_d$, $b_1\geq\dotsm \geq b_d$, and $c_1\geq\dotsm\geq c_d$, respectively. Then,
	\begin{align}
	\sum_{i=1}^d p_{i,d-i+1}(c_i-c_{d-i+1})^2\leq \sum_{i=1}^d p_{i,d-i+1}\left(a_i-a_{d-i+1}+b_i-b_{d-i+1}\right)^2.
	\end{align}
\end{lemma}
\begin{proof}
	Let us first show that coefficients $p_{i,d-i+1}$ satisfy 
	\begin{align}
	(p_{i,d-i+1})_{i=1}^{\ceil*{d/2}}\in\mathcal{D}_+^{\ceil*{d/2}}.\label{eq:in_D}
	\end{align}
	For $1\leq i<\floor*{d/2}$, where $\floor*{d/2}$ denotes the largest integer $j$ with $j\leq d/2$, we have
	\begin{align}
	p_{i,d-i+1}\geq p_{i,i+1}+p_{i+1,d-i+1}\geq p_{i,i+1}+p_{i+1,d-i}+p_{d-i,d-i+1},
	\end{align}
	where inequality $(i)$ from lemma \ref{lem:1} was applied twice, and it follows	$p_{i,d-i+1}\geq p_{i+1,d-i}$. For even $d$ it follows Eq.~\eqref{eq:in_D}.
	For odd $d$, we further have $p_{\floor*{d/2},d-\floor*{d/2}+1}\geq p_{\ceil*{d/2},d-\ceil*{d/2}+1}$ because $p_{\ceil*{d/2},d-\ceil*{d/2}+1}=p_{\ceil*{d/2},\ceil*{d/2}}=0$ by definition of $p_{i,j}$. This proves Eq.~\eqref{eq:in_D}.

	Together with lemmata \ref{lem:weakmajo} and \ref{lem:Schurconvex} it follows that
	\begin{align}\label{eq:lowersum}
	\sum_{i=1}^{\ceil*{d/2}} p_{i,d-i+1}(c_i-c_{d-i+1})^2\leq \sum_{i=1}^{\ceil*{d/2}} p_{i,d-i+1}\left(a_i-a_{d-i+1}+b_i-b_{d-i+1}\right)^2,
	\end{align}
	which, due to the symmetries $p_{i,d-i+1}=p_{d-i+1,i}$ and $(c_i-c_j)^2=(c_j-c_i)^2$, is equivalent to
	\begin{align}\label{eq:uppersum}
	\sum_{i=d-\ceil*{d/2}+1}^{d} p_{i,d-i+1}(c_i-c_{d-i+1})^2\leq \sum_{i=d-\ceil*{d/2}+1}^{d} p_{i,d-i+1}\left(a_i-a_{d-i+1}+b_i-b_{d-i+1}\right)^2.
	\end{align}
	Adding inequalities \eqref{eq:lowersum} and \eqref{eq:uppersum} proves the lemma since, in case of odd $d$, $p_{\ceil*{d/2},\ceil*{d/2}}=0$.
\end{proof}
We are now in the position to prove theorem 2 from the Letter:
\begin{theorem}\label{th:time_dep_sup}
	For any state $\rho$ with ordered eigenvalues $p_1\geq \dotsm \geq p_d$ and any time-dependent Hamiltonian $H_\alpha(t)$, where $\mu_1(t)\geq\dotsm\geq \mu_d(t)$ are the ordered eigenvalues of $\partial_\alpha H_\alpha(t)\coloneqq\partial H_\alpha(t)/\partial \alpha$, an upper bound for the QFI is given by
	\begin{align}\label{eq:opt_QFI_th2_sup}
	K_\alpha=\frac{1}{2}\sum_{k=1}^{d}p_{k,d-k+1} \left(\int_{0}^{T}[\mu_k(t)-\mu_{d-k+1}(t)]\emph{d}t\right)^2.
	\end{align}
	Let $\ket{\mu_k(t)}$ be the time-dependent eigenvectors of $\partial_\alpha H_\alpha(t)$, $\partial_\alpha H_\alpha(t)\ket{\mu_k(t)}=\mu_k(t)\ket{\mu_k(t)}$. The upper bound $K_\alpha$ is reached by preparing the initial state
	\begin{equation}\label{eq:opt_state_sup_2}
	\rho^*=\sum_{k=1}^d p_k \ket{\phi_k}\bra{\phi_k},
	\end{equation}
	with
	\begin{equation}\label{eq:opt_eigen_th2_sup}
	\ket{\phi_k}=\begin{cases*}
	\frac{\ket{\mu_k(0)}+e^{\ii\chi_k}\ket{\mu_{d-k+1}(0)}}{\sqrt{2}}\quad ~~\text{if } 2k<d+1, \\
	\mathrlap{\ket{\mu_k(0)}}\hphantom{\frac{\ket{\mu_k(0)}+\ket{\mu_{d-k+1}(0)}}{\sqrt{2}}}\quad ~~~~~~\text{if }2k=d+1,\\
	\frac{\ket{\mu_k(0)}-e^{\ii\chi_k}\ket{\mu_{d-k+1}(0)}}{\sqrt{2}}\quad ~~\text{if } 2k>d+1,
	\end{cases*}
	\end{equation}
	where $\chi_k$ are arbitrary real phases (the theorem as formulated in the Letter is recovered by setting $\chi_k=0$),  and by choosing the Hamiltonian control $H_c(t)$ such that
	\begin{align}
	U_\alpha(t)\ket{\mu_k(0)}=\ket{\mu_k(t)}\quad \forall k=1,\dotsc,d~~\forall t,\label{eq:opt_con_th2_sup}
	\end{align}
	where
	\begin{align}\label{eq:unitarycon_sup}
	U_\alpha(t)= \mc{T}\left[\exp\left(-\ii\int_0^t [H_\alpha(\tau)+H_c(\tau)] \diff \tau\right)\right].
	\end{align}
\end{theorem}
\begin{proof}
	From theorem \ref{th:main_sup} we have that
	for any state $\rho$ and any generator $\dgen$ with ordered eigenvalues $p_1\geq \dotsm \geq p_d$ and $h_1\geq\dotsm\geq h_d$, respectively, the maximal QFI with respect to all unitary state preparations $U\rho U^\dagger$, $U\in\text{U}(d)$, is given by
	\begin{align}
	I_\alpha^*\coloneqq\max_{U} \qfinkd_\param (U\rho U^\dagger)=\frac{1}{2}\sum_{j=1}^{d}\pij{j}{d-j+1}(h_j-h_{d-j+1})^2.
	\end{align}
	Further, the generator can be written as \cite[Eq.~6]{pang_optimal_2017}	
	\begin{align}\label{eq:generator_sup}
	h_\alpha=\int_{0}^{T}U_\alpha^\dagger(t)\partial_\alpha H_\alpha(t)U_\alpha(t) \diff t,
	\end{align}
	Writing the integral as an infinite sum,
	\begin{align}
	h_\alpha=\lim\limits_{n\rightarrow \infty}\sum_{l=0}^{n}U_\alpha^\dagger(lT/n)\partial_\alpha H_\alpha(lT/n)U_\alpha(lT/n)T/n,\label{eq:generator_sum_sup}
	\end{align}
	repeated application of lemma \ref{cor:triangle} to bipartitions of the sum yields in the limit of infinite many applications of lemma \ref{cor:triangle}
	\begin{align}\label{eq:optQFI_sup}
	I^*_\alpha=\frac{1}{2}\sum_{j=1}^{d}\pij{j}{d-j+1}(h_j-h_{d-j+1})^2\leq \frac{1}{2}\sum_{j=1}^{d}\pij{j}{d-j+1}\left(\int_0^T[\mu_j(t)-\mu_{d-j+1}(t)]\diff t\right)^2=K_\alpha .
	\end{align}
	It remains to show that Eq.~\eqref{eq:opt_QFI_th2_sup} can be saturated. In order to show this it suffices to calculate the QFI for the initial state as defined in Eqs.\,\eqref{eq:opt_state_sup_2} and \eqref{eq:opt_eigen_th2_sup} and a generator as given in Eq.~\eqref{eq:generator_sum_sup} with the unitary transformation fulfilling Eq.~\eqref{eq:unitarycon_sup}:
	\begin{align}
	I_\param (\rho^*)&=2\sum_{i,j=1}^d \pij{i}{j} \absv{\bra{\phi_i}\dgen\ket{\phi_j}}^2
	\end{align}
	where $\ket{\phi_j}$ are defined in Eq.~\eqref{eq:opt_eigen_th2_sup}. More explicitly,  in 
	\begin{align}
	\bra{\phi_i}\dgen\ket{\phi_j}&=\lim\limits_{n\rightarrow \infty}\sum_{l=0}^{n}\bra{\phi_i}U_\alpha^\dagger(lT/n)\partial_\alpha H_\alpha(lT/n)U_\alpha(lT/n)\ket{\phi_j}T/n\label{eq:matrixcoef}
	\end{align}
	we use the definition of $\ket{\phi_j}$ and Eq.~\eqref{eq:opt_con_th2_sup} which gives, due to
	\begin{align}
	\bra{\mu_i(lT/n)}\partial_\alpha H_\alpha(lT/n)\ket{\mu_j(lT/n)}=\delta_{i,j}\mu_i(lT/n),
	\end{align}
	the following expression for the matrix coefficients in Eq.~\eqref{eq:matrixcoef}:
	\begin{multline}
	\bra{\phi_i}U_\alpha^\dagger(lT/n)\partial_\alpha H_\alpha(lT/n)U_\alpha(lT/n)\ket{\phi_j}\\=\delta_{i,j}\frac{\mu_i(lT/n)+\mu_{d-i+1}(lT/n)}{2} + \delta_{i,d-j+1}\frac{|\mu_i(lT/n)-\mu_{d-i+1}(lT/n)|}{2}.
	\end{multline}
	Due to $p_{i,i}=0$ one obtains
	\begin{align}
	I_\param (\rho^*)&=2\sum_{j=1}^d \pij{j}{d-j+1}\left|\lim\limits_{n\rightarrow \infty}\sum_{l=0}^{n}\frac{\mu_j(lT/n)-\mu_{d-j+1}(lT/n)}{2}T/n\right|^2\\
	&=\frac{1}{2}\sum_{j=1}^{d}p_{j,d-j+1} \left(\int_{0}^{T}[\mu_j(t)-\mu_{d-j+1}(t)]\diff t\right)^2=K_\alpha.
	\end{align}
\end{proof}
\section{Proof of Heisenberg scaling for thermal states}
In this section we will prove that if a product of $N$ thermal spin-$j$ states (at arbitrary finite temperature) is available and sensor dynamics is unitary, one can reach Heisenberg scaling of the QFI $I_\alpha$ for unitary dynamics in $N$ and $j$ by preparing the optimal initial state according theorem 1 in the Letter (or theorem 2, in case of Hamiltonian control). Heisenberg scaling in $N$ and $j$ means $I_\alpha\propto N^2$ for any $j=\frac{1}{2},1,\frac{3}{2},\dotsc$ and $I_\alpha\propto j^2$ for any $N=1,2,3,\dotsc$.

According to the pinching theorem (also known as squeeze theorem) a function scales with $N^2$ ($j^2$) if there are upper and lower bounds scaling as $N^2$ ($j^2$). Clearly, the QFI of a product of $N$ thermal spin-$j$ states is upper bounded by the pure-state case obtained in the limiting case of zero temperature. For pure states, it is well known that the QFI, optimized over unitary state preparations, scales as $N^2$ ($j^2$). We will find lower bounds for the QFI of a product of $N$ thermal spin-$j$ states that scale as $N^2$ ($j^2$).

Let the QFI be given by (compared to Eq.~(14) in the Letter, we set $g(T)=1$ because we are only interested in the scaling with $N$ and $j$ in the following)
\begin{align}
K_B=\frac{4}{Z_\beta^N}\sum_{k=-Nj}^{Nj}q(k)\frac{\sinh^2(\beta k)}{\cosh(\beta k)}k^2,\label{eq:qfi_opt_spins_sup}
\end{align}
with $q(k)$
the number of possibilities of getting a sum $k$ when rolling
$N$ fair dice,
each having $2j+1$ sides corresponding to values $\{-j,\dotsc, j\}$, and with $Z_\beta$ the partition function
\begin{align}
Z_\beta&=\sum_{m=-j}^{j}\text{e}^{\beta m}=\cosh(\beta j) + \frac{\sinh(\beta j)}{\tanh(\beta/2)},\label{eq:partfunc}
\end{align}
which was rewritten (for $\beta>0$) making use of the geometric series.
First, we find a lower bound $L_B$ for $K_B$:
\begin{align}
K_B=\frac{4}{Z_\beta^N}\sum_{k=-Nj}^{Nj}q(k)\frac{\sinh^2(\beta k)}{\cosh(\beta k)}k^2&=\frac{4}{Z_\beta^N}\sum_{k=-Nj}^{Nj}q(k)\left(\cosh(\beta k)-\frac{1}{\cosh(\beta k)}\right)k^2\label{eq:qfi_two_terms}\\
&\geq \frac{4}{Z_\beta^N}\sum_{k=-Nj}^{Nj}q(k)\left(\cosh(\beta k)-1\right)k^2\eqqcolon L_B,
\end{align}
where we used that each summand is nonnegative and 
\begin{align}
\frac{\sinh^2(\beta k)}{\cosh(\beta k)}&=\frac{\cosh^2(\beta k)-1}{\cosh(\beta k)}=\cosh(\beta k)-\frac{1}{\cosh(\beta k)}
\geq\cosh(\beta k)-1,
\end{align}
which follows from the trigonometric identity $\sinh^2(x)-\cosh^2(x)=1$ and $\cosh(x)\geq 1$. Next, we rewrite $L_B$ as
\begin{align}
L_B=\frac{4}{Z_\beta^N}\sum_{k=-Nj}^{Nj}q(k)\left(\ee^{\beta k}-1\right)k^2,\label{eq:lb1}
\end{align}
where we used that $\cosh(x)=(\ee^x+\ee^{-x})/2$ and $\sum_{k=-Nj}^{Nj}q(k)\ee^{\beta k}k^2=\sum_{k=-Nj}^{Nj}q(k)\ee^{-\beta k}k^2$ because $q(k)k^2$ is symmetric around $k=0$.

Then, we make use of the  generating function of $q(k)$ \cite{uspensky1937introduction}:
\begin{align}
\left(x^{-j}+x^{-j+1}+\dotsm + x^{j}\right)^N=\sum_{k=-Nj}^{Nj}q(k)x^k.
\end{align}
By setting $x=\ee^\beta$, we find $Z_\beta^N=\sum_{k=-Nj}^{Nj}q(k)\ee^{\beta k}$. Taking the second derivative with respect to $\beta$ yields
\begin{align}
\frac{\partial Z_\beta^N}{\partial\beta^2}=\sum_{k=-Nj}^{Nj}q(k)\ee^{\beta k}k^2.
\end{align}
With this, we rewrite $L_B$ as
\begin{align}
L_B=\left(\frac{\partial^2 Z_\beta^N}{\partial\beta^2}-\left.\frac{\partial^2 Z_\beta^N}{\partial\beta^2}\right|_{\beta=0}\right)\frac{4}{Z_\beta^N},\label{eq:lb2}
\end{align}
where the second term is evaluated for $\beta=0$ and corresponds the negative part of $L_B$ in Eq.~\eqref{eq:lb1}.
Since $Z_\beta\geq Z_{\beta=0}$, we find the lower bound
\begin{align}
-\left.\frac{\partial^2 Z_\beta^N}{\partial\beta^2}\right|_{\beta=0}\frac{4}{Z_\beta^N}\geq -\left.\frac{\partial^2 Z_\beta^N}{\partial\beta^2}\right|_{\beta=0}\frac{4}{Z_0^N},
\end{align}
which is readily evaluated:
\begin{align}
-\left.\frac{\partial^2 Z_\beta^N}{\partial\beta^2}\right|_{\beta=0}\frac{4}{Z_0^N}=-4\left[N(N-1)\left(\frac{\partial Z_\beta/\partial\beta|_{\beta=0}}{Z_0}\right)^2+N\frac{\partial^2 Z_\beta/\partial\beta^2|_{\beta=0}}{Z_0}\right],
\end{align}
and with $\partial Z_\beta/\partial\beta|_{\beta=0}=\sum_{m=-j}^{j}m=0$ we find
\begin{align}
-\left.\frac{\partial^2 Z_\beta^N}{\partial\beta^2}\right|_{\beta=0}\frac{4}{Z_0^N}=-4N\sum_{m=-j}^{j}m^2/(2j+1)=-\frac{4}{3}Nj(j+1),\label{eq:betazero}
\end{align}	
again using the geometric series. Together with
\begin{align}
\frac{\partial^2 Z_\beta^N}{\partial\beta^2}\frac{4}{Z_\beta^N}=4N(N-1)\left(\frac{\partial Z_\beta/\partial\beta}{Z_\beta}\right)^2+4N\frac{\partial^2Z_\beta/\partial\beta^2}{Z_\beta},
\end{align}
this brings us to another lower bound:  
\begin{align}
L_B\geq M_B\coloneqq 4\left[N(N-1)\left(\frac{\partial Z_\beta/\partial\beta}{Z_\beta}\right)^2+N\frac{\partial^2Z_\beta/\partial\beta^2}{Z_\beta}-\frac{1}{3}Nj(j+1)\right]
\end{align}
Neglecting terms proportional to $N$ we find after trivial algebraic transformations
\begin{align}
M_B&\propto 4N^2\left(\frac{\partial Z_\beta/\partial\beta}{Z_\beta}\right)^2\label{eq:Mb}\\
&=N^2\frac{\left\{(1 + j) \sinh(\beta j) - j \sinh[\beta (1 + j)]\right\}^2}{\sinh^2\left(\frac{\beta}{2}\right)\sinh^2\left[\beta(j+\frac{1}{2})\right]},\label{eq:n2scaling}
\end{align}
which is clearly nonnegative for finite temperatures ($\beta>0$). In order to become zero,
\begin{align}
\frac{\sinh[\beta j]}{\beta j}=\frac{\sinh[\beta(j+1)]}{\beta(j+1)}\label{eq:contr}
\end{align}
would have to be fulfilled. However, since $\frac{\sinh(x)}{x}=1+\frac{x^2}{3!}+\frac{x^4}{5!}+\dotsm$ is an increasing function on $x\in[0,\infty)$, Eq.~\eqref{eq:contr} leads to a contradiction for $\beta>0$. Therefore, the expression in Eq.~\eqref{eq:n2scaling} is positive for finite temperatures which proves the $N^2$ scaling for all $j=\frac{1}{2},1,\frac{3}{2},\dotsc$.

There are two remarks in order:
\begin{enumerate}[(i)]
	\item In summary the lower bound $M_B$ was obtained from QFI $K_B$ by finding a lower bound for the negative term of the QFI in Eq.~\eqref{eq:qfi_two_terms},
	\begin{align}
	-\frac{4}{Z_\beta^N}\sum_{k=-Nj}^{Nj}q(k)\frac{k^2}{\cosh(\beta k)}\geq -\frac{4}{3}Nj(j+1).
	\end{align}
	Since the right-hand side scales linearly in $N$, the left-hand side scales at most linearly in $N$ and, in particular, cannot scale quadratically. Thus, the $N^2$-scaling of the QFI solely comes from the positive term in Eq.~\eqref{eq:qfi_two_terms}. Therefore, in leading order of $N$ we find for the QFI exactly Eq.~\eqref{eq:Mb}, i.e.,
	\begin{align}
	K_B=4N^2\left(\frac{\partial Z_\beta/\partial\beta}{Z_\beta}\right)^2+\mathcal{O}(N)=4N^2\left(\frac{\partial \ln Z_\beta}{\partial \beta}\right)^2+\mathcal{O}(N),
	\end{align}
	where $\mathcal{O}(N)$ denotes terms $\propto N$ and lower-order terms.
	With the operator $S_z$ of a spin $j$ in $z$-direction and corresponding thermal state $\rho_\text{th}=\text{e}^{\beta S_z}/Z_\beta$, we rewrite $Z_\beta=\trb{\text{e}^{\beta S_z}}$ and
	\begin{align}
	4N^2\left(\frac{\partial \ln Z_\beta}{\partial \beta}\right)^2=4N^2\braket{S_z}^2,\label{eq:nsquared}
	\end{align}
	where $\braket{S_z}=\trb{\rho_\text{th}S_z}$.
	\item Let us identify $\beta=1/\chi$ with a temperature $\chi$. A Taylor expansion of the prefactor $Q(\beta)\coloneqq 4\left(\frac{\partial \ln Z_\beta}{\partial \beta}\right)^2$ in Eq.~\eqref{eq:nsquared} around $\beta=0$ yields
	\begin{align}
	Q(\beta)=\frac{4}{9}\left[j(j+1)\right]^2\beta^2+\mathcal{O}(\beta^4),\label{eq:taylor}
	\end{align}
	where $\mathcal{O}(\beta^4)$ denotes terms $\propto\beta^4$ and higher-order terms which can be neglected for small $\beta$. This shows that for small $\beta$, i.e., for large temperatures $\chi$, the prefactor decays quadratically, $\propto \chi^{-2}$. Also, in this regime of high temperatures, the QFI scales with $j^4$. However, if the product $j\beta$ is of order one (or larger), we are no longer in the range of validity of the second-order Taylor expansion in Eq.~\eqref{eq:taylor}. As we will see in the next section, the QFI scales $\propto j^2$ in the limit of large $j$.
\end{enumerate}
In order to prove $j^2$ scaling, we first consider
\begin{align}
M_B(N)-NM_B(N=1)=N(N-1)\left(\frac{\partial Z_\beta/\partial\beta}{Z_\beta}\right)^2\geq 0
\end{align}
for any $N\geq 1$. It follows that $M_B(N)\geq NM_B(N=1)$. We find after some simple algebraic transformations using the partition function as given in Eq.~\eqref{eq:partfunc},
\begin{align}
N M_B(1)&=4N\left[\frac{\partial^2Z_\beta/\partial\beta^2}{Z_\beta}-\frac{j(j+1)}{3}\right]\\
&=N\left[j^2\frac{8}{3}-j\frac{10 \cosh(\beta j) + 2\cosh[\beta (j+1)]}{3\sinh(\beta/2)\sinh[\beta(j+1/2)]}+\frac{\sinh(\beta)\sinh(\beta j)}{\sinh^3(\beta/2)\sinh[\beta(j+1/2)]}\right],
\end{align}
which is well defined for finite temperatures, and the third summand as well as the prefactor of $j$ in the second summand clearly converge to a constant in the limit of large $j$. This proves the $j^2$ scaling for finite temperatures for any $N=1,2,\dotsc$.


\begin{thebibliography}{54}%
	\makeatletter
	\providecommand \@ifxundefined [1]{%
		\@ifx{#1\undefined}
	}%
	\providecommand \@ifnum [1]{%
		\ifnum #1\expandafter \@firstoftwo
		\else \expandafter \@secondoftwo
		\fi
	}%
	\providecommand \@ifx [1]{%
		\ifx #1\expandafter \@firstoftwo
		\else \expandafter \@secondoftwo
		\fi
	}%
	\providecommand \natexlab [1]{#1}%
	\providecommand \enquote  [1]{``#1''}%
	\providecommand \bibnamefont  [1]{#1}%
	\providecommand \bibfnamefont [1]{#1}%
	\providecommand \citenamefont [1]{#1}%
	\providecommand \href@noop [0]{\@secondoftwo}%
	\providecommand \href [0]{\begingroup \@sanitize@url \@href}%
	\providecommand \@href[1]{\@@startlink{#1}\@@href}%
	\providecommand \@@href[1]{\endgroup#1\@@endlink}%
	\providecommand \@sanitize@url [0]{\catcode `\\12\catcode `\$12\catcode
		`\&12\catcode `\#12\catcode `\^12\catcode `\_12\catcode `\%12\relax}%
	\providecommand \@@startlink[1]{}%
	\providecommand \@@endlink[0]{}%
	\providecommand \url  [0]{\begingroup\@sanitize@url \@url }%
	\providecommand \@url [1]{\endgroup\@href {#1}{\urlprefix }}%
	\providecommand \urlprefix  [0]{URL }%
	\providecommand \Eprint [0]{\href }%
	\providecommand \doibase [0]{https://doi.org/}%
	\providecommand \selectlanguage [0]{\@gobble}%
	\providecommand \bibinfo  [0]{\@secondoftwo}%
	\providecommand \bibfield  [0]{\@secondoftwo}%
	\providecommand \translation [1]{[#1]}%
	\providecommand \BibitemOpen [0]{}%
	\providecommand \bibitemStop [0]{}%
	\providecommand \bibitemNoStop [0]{.\EOS\space}%
	\providecommand \EOS [0]{\spacefactor3000\relax}%
	\providecommand \BibitemShut  [1]{\csname bibitem#1\endcsname}%
	\let\auto@bib@innerbib\@empty
	\bibitem [{\citenamefont {Giovannetti}\ \emph {et~al.}(2004)\citenamefont
		{Giovannetti}, \citenamefont {Lloyd},\ and\ \citenamefont
		{Maccone}}]{giovannetti_quantum_2004}%
	\BibitemOpen
	\bibfield  {author} {\bibinfo {author} {\bibfnamefont {V.}~\bibnamefont
			{Giovannetti}}, \bibinfo {author} {\bibfnamefont {S.}~\bibnamefont {Lloyd}},\
		and\ \bibinfo {author} {\bibfnamefont {L.}~\bibnamefont {Maccone}},\
	}\bibfield  {title} {\bibinfo {title} {{Quantum}-{Enhanced} {Measurements}:
			{Beating} the {Standard} {Quantum} {Limit}},\ }\href
	{https://doi.org/10.1126/science.1104149} {\bibfield  {journal} {\bibinfo
			{journal} {Science}\ }\textbf {\bibinfo {volume} {306}},\ \bibinfo {pages}
		{1330} (\bibinfo {year} {2004})}\BibitemShut {NoStop}%
	\bibitem [{\citenamefont {Paris}(2009)}]{paris_quantum_2009}%
	\BibitemOpen
	\bibfield  {author} {\bibinfo {author} {\bibfnamefont {M.~G.~A.}\
			\bibnamefont {Paris}},\ }\bibfield  {title} {\bibinfo {title} {Quantum
			estimation for quantum technology},\ }\href
	{https://doi.org/10.1142/S0219749909004839} {\bibfield  {journal} {\bibinfo
			{journal} {International Journal of Quantum Information}\ }\textbf {\bibinfo
			{volume} {07}},\ \bibinfo {pages} {125} (\bibinfo {year} {2009})}\BibitemShut
	{NoStop}%
	\bibitem [{\citenamefont {Giovannetti}\ \emph {et~al.}(2011)\citenamefont
		{Giovannetti}, \citenamefont {Lloyd},\ and\ \citenamefont
		{Maccone}}]{giovannetti_advances_2011}%
	\BibitemOpen
	\bibfield  {author} {\bibinfo {author} {\bibfnamefont {V.}~\bibnamefont
			{Giovannetti}}, \bibinfo {author} {\bibfnamefont {S.}~\bibnamefont {Lloyd}},\
		and\ \bibinfo {author} {\bibfnamefont {L.}~\bibnamefont {Maccone}},\
	}\bibfield  {title} {\bibinfo {title} {Advances in quantum metrology},\
	}\href {https://doi.org/10.1038/nphoton.2011.35} {\bibfield  {journal}
		{\bibinfo  {journal} {Nature Photonics}\ }\textbf {\bibinfo {volume} {5}},\
		\bibinfo {pages} {222} (\bibinfo {year} {2011})}\BibitemShut {NoStop}%
	\bibitem [{\citenamefont {Pezz\`e}\ \emph {et~al.}(2018)\citenamefont
		{Pezz\`e}, \citenamefont {Smerzi}, \citenamefont {Oberthaler}, \citenamefont
		{Schmied},\ and\ \citenamefont {Treutlein}}]{pezze2018quantum}%
	\BibitemOpen
	\bibfield  {author} {\bibinfo {author} {\bibfnamefont {L.}~\bibnamefont
			{Pezz\`e}}, \bibinfo {author} {\bibfnamefont {A.}~\bibnamefont {Smerzi}},
		\bibinfo {author} {\bibfnamefont {M.~K.}\ \bibnamefont {Oberthaler}},
		\bibinfo {author} {\bibfnamefont {R.}~\bibnamefont {Schmied}},\ and\ \bibinfo
		{author} {\bibfnamefont {P.}~\bibnamefont {Treutlein}},\ }\bibfield  {title}
	{\bibinfo {title} {Quantum metrology with nonclassical states of atomic
			ensembles},\ }\href {https://doi.org/10.1103/RevModPhys.90.035005} {\bibfield
		{journal} {\bibinfo  {journal} {Rev. Mod. Phys.}\ }\textbf {\bibinfo
			{volume} {90}},\ \bibinfo {pages} {035005} (\bibinfo {year}
		{2018})}\BibitemShut {NoStop}%
	\bibitem [{\citenamefont {Braun}\ \emph {et~al.}(2018)\citenamefont {Braun},
		\citenamefont {Adesso}, \citenamefont {Benatti}, \citenamefont {Floreanini},
		\citenamefont {Marzolino}, \citenamefont {Mitchell},\ and\ \citenamefont
		{Pirandola}}]{braun2018quantum}%
	\BibitemOpen
	\bibfield  {author} {\bibinfo {author} {\bibfnamefont {D.}~\bibnamefont
			{Braun}}, \bibinfo {author} {\bibfnamefont {G.}~\bibnamefont {Adesso}},
		\bibinfo {author} {\bibfnamefont {F.}~\bibnamefont {Benatti}}, \bibinfo
		{author} {\bibfnamefont {R.}~\bibnamefont {Floreanini}}, \bibinfo {author}
		{\bibfnamefont {U.}~\bibnamefont {Marzolino}}, \bibinfo {author}
		{\bibfnamefont {M.~W.}\ \bibnamefont {Mitchell}},\ and\ \bibinfo {author}
		{\bibfnamefont {S.}~\bibnamefont {Pirandola}},\ }\bibfield  {title} {\bibinfo
		{title} {Quantum-enhanced measurements without entanglement},\ }\href
	{https://doi.org/10.1103/RevModPhys.90.035006} {\bibfield  {journal}
		{\bibinfo  {journal} {Rev. Mod. Phys.}\ }\textbf {\bibinfo {volume} {90}},\
		\bibinfo {pages} {035006} (\bibinfo {year} {2018})}\BibitemShut {NoStop}%
	\bibitem [{\citenamefont {Helstrom}(1976)}]{helstrom_quantum_1976}%
	\BibitemOpen
	\bibfield  {author} {\bibinfo {author} {\bibfnamefont {C.~W.}\ \bibnamefont
			{Helstrom}},\ }\href@noop {} {\emph {\bibinfo {title} {{Quantum {Detection}
					and {Estimation} {Theory}}}}},\ \bibinfo {series} {Mathematics in {Science}
		and {Engineering}}, Vol.\ \bibinfo {volume} {123}\ (\bibinfo  {publisher}
	{Elsevier},\ \bibinfo {year} {1976})\BibitemShut {NoStop}%
	\bibitem [{\citenamefont {Braunstein}\ and\ \citenamefont
		{Caves}(1994)}]{braunstein_statistical_1994}%
	\BibitemOpen
	\bibfield  {author} {\bibinfo {author} {\bibfnamefont {S.~L.}\ \bibnamefont
			{Braunstein}}\ and\ \bibinfo {author} {\bibfnamefont {C.~M.}\ \bibnamefont
			{Caves}},\ }\bibfield  {title} {\bibinfo {title} {Statistical distance and
			the geometry of quantum states},\ }\href
	{https://doi.org/10.1103/PhysRevLett.72.3439} {\bibfield  {journal} {\bibinfo
			{journal} {Phys. Rev. Lett.}\ }\textbf {\bibinfo {volume} {72}},\ \bibinfo
		{pages} {3439} (\bibinfo {year} {1994})}\BibitemShut {NoStop}%
	\bibitem [{\citenamefont {Fernholz}\ \emph {et~al.}(2008)\citenamefont
		{Fernholz}, \citenamefont {Krauter}, \citenamefont {Jensen}, \citenamefont
		{Sherson}, \citenamefont {S\o{}rensen},\ and\ \citenamefont
		{Polzik}}]{fernholz2008spin}%
	\BibitemOpen
	\bibfield  {author} {\bibinfo {author} {\bibfnamefont {T.}~\bibnamefont
			{Fernholz}}, \bibinfo {author} {\bibfnamefont {H.}~\bibnamefont {Krauter}},
		\bibinfo {author} {\bibfnamefont {K.}~\bibnamefont {Jensen}}, \bibinfo
		{author} {\bibfnamefont {J.~F.}\ \bibnamefont {Sherson}}, \bibinfo {author}
		{\bibfnamefont {A.~S.}\ \bibnamefont {S\o{}rensen}},\ and\ \bibinfo {author}
		{\bibfnamefont {E.~S.}\ \bibnamefont {Polzik}},\ }\bibfield  {title}
	{\bibinfo {title} {{Spin Squeezing of Atomic Ensembles via Nuclear-Electronic
				Spin Entanglement}},\ }\href {https://doi.org/10.1103/PhysRevLett.101.073601}
	{\bibfield  {journal} {\bibinfo  {journal} {Phys. Rev. Lett.}\ }\textbf
		{\bibinfo {volume} {101}},\ \bibinfo {pages} {073601} (\bibinfo {year}
		{2008})}\BibitemShut {NoStop}%
	\bibitem [{\citenamefont {Andr\'e}\ and\ \citenamefont
		{Lukin}(2002)}]{andre2002atom}%
	\BibitemOpen
	\bibfield  {author} {\bibinfo {author} {\bibfnamefont {A.}~\bibnamefont
			{Andr\'e}}\ and\ \bibinfo {author} {\bibfnamefont {M.~D.}\ \bibnamefont
			{Lukin}},\ }\bibfield  {title} {\bibinfo {title} {Atom correlations and spin
			squeezing near the {Heisenberg} limit: {Finite-size} effect and
			decoherence},\ }\href {https://doi.org/10.1103/PhysRevA.65.053819} {\bibfield
		{journal} {\bibinfo  {journal} {Phys. Rev. A}\ }\textbf {\bibinfo {volume}
			{65}},\ \bibinfo {pages} {053819} (\bibinfo {year} {2002})}\BibitemShut
	{NoStop}%
	\bibitem [{\citenamefont {Leroux}\ \emph {et~al.}(2010)\citenamefont {Leroux},
		\citenamefont {Schleier-Smith},\ and\ \citenamefont
		{Vuleti\ifmmode~\acute{c}\else \'{c}\fi{}}}]{leroux2010implementation}%
	\BibitemOpen
	\bibfield  {author} {\bibinfo {author} {\bibfnamefont {I.~D.}\ \bibnamefont
			{Leroux}}, \bibinfo {author} {\bibfnamefont {M.~H.}\ \bibnamefont
			{Schleier-Smith}},\ and\ \bibinfo {author} {\bibfnamefont {V.}~\bibnamefont
			{Vuleti\ifmmode~\acute{c}\else \'{c}\fi{}}},\ }\bibfield  {title} {\bibinfo
		{title} {{Implementation of Cavity Squeezing of a Collective Atomic Spin}},\
	}\href {https://doi.org/10.1103/PhysRevLett.104.073602} {\bibfield  {journal}
		{\bibinfo  {journal} {Phys. Rev. Lett.}\ }\textbf {\bibinfo {volume} {104}},\
		\bibinfo {pages} {073602} (\bibinfo {year} {2010})}\BibitemShut {NoStop}%
	\bibitem [{\citenamefont {Orzel}\ \emph {et~al.}(2001)\citenamefont {Orzel},
		\citenamefont {Tuchman}, \citenamefont {Fenselau}, \citenamefont {Yasuda},\
		and\ \citenamefont {Kasevich}}]{orzel2001squeezed}%
	\BibitemOpen
	\bibfield  {author} {\bibinfo {author} {\bibfnamefont {C.}~\bibnamefont
			{Orzel}}, \bibinfo {author} {\bibfnamefont {A.~K.}\ \bibnamefont {Tuchman}},
		\bibinfo {author} {\bibfnamefont {M.~L.}\ \bibnamefont {Fenselau}}, \bibinfo
		{author} {\bibfnamefont {M.}~\bibnamefont {Yasuda}},\ and\ \bibinfo {author}
		{\bibfnamefont {M.~A.}\ \bibnamefont {Kasevich}},\ }\bibfield  {title}
	{\bibinfo {title} {{Squeezed States in a Bose-Einstein Condensate}},\ }\href
	{https://doi.org/10.1126/science.1058149} {\bibfield  {journal} {\bibinfo
			{journal} {Science}\ }\textbf {\bibinfo {volume} {291}},\ \bibinfo {pages}
		{2386} (\bibinfo {year} {2001})}\BibitemShut {NoStop}%
	\bibitem [{\citenamefont {Giovannetti}\ \emph {et~al.}(2006)\citenamefont
		{Giovannetti}, \citenamefont {Lloyd},\ and\ \citenamefont
		{Maccone}}]{giovannetti_quantum_2006}%
	\BibitemOpen
	\bibfield  {author} {\bibinfo {author} {\bibfnamefont {V.}~\bibnamefont
			{Giovannetti}}, \bibinfo {author} {\bibfnamefont {S.}~\bibnamefont {Lloyd}},\
		and\ \bibinfo {author} {\bibfnamefont {L.}~\bibnamefont {Maccone}},\
	}\bibfield  {title} {\bibinfo {title} {{Quantum Metrology}},\ }\href
	{https://doi.org/10.1103/PhysRevLett.96.010401} {\bibfield  {journal}
		{\bibinfo  {journal} {Phys. Rev. Lett.}\ }\textbf {\bibinfo {volume} {96}},\
		\bibinfo {pages} {010401} (\bibinfo {year} {2006})}\BibitemShut {NoStop}%
	\bibitem [{\citenamefont {Fujiwara}\ and\ \citenamefont
		{Imai}(2008)}]{fujiwara_fibre_2008}%
	\BibitemOpen
	\bibfield  {author} {\bibinfo {author} {\bibfnamefont {A.}~\bibnamefont
			{Fujiwara}}\ and\ \bibinfo {author} {\bibfnamefont {H.}~\bibnamefont
			{Imai}},\ }\bibfield  {title} {\bibinfo {title} {A fibre bundle over
			manifolds of quantum channels and its application to quantum statistics},\
	}\href {https://doi.org/10.1088/1751-8113/41/25/255304} {\bibfield  {journal}
		{\bibinfo  {journal} {Journal of Physics A: Mathematical and Theoretical}\
		}\textbf {\bibinfo {volume} {41}},\ \bibinfo {pages} {255304} (\bibinfo
		{year} {2008})}\BibitemShut {NoStop}%
	\bibitem [{\citenamefont {Modi}\ \emph {et~al.}(2011)\citenamefont {Modi},
		\citenamefont {Cable}, \citenamefont {Williamson},\ and\ \citenamefont
		{Vedral}}]{modi2011quantum}%
	\BibitemOpen
	\bibfield  {author} {\bibinfo {author} {\bibfnamefont {K.}~\bibnamefont
			{Modi}}, \bibinfo {author} {\bibfnamefont {H.}~\bibnamefont {Cable}},
		\bibinfo {author} {\bibfnamefont {M.}~\bibnamefont {Williamson}},\ and\
		\bibinfo {author} {\bibfnamefont {V.}~\bibnamefont {Vedral}},\ }\bibfield
	{title} {\bibinfo {title} {{Quantum Correlations in Mixed-State Metrology}},\
	}\href {https://doi.org/10.1103/PhysRevX.1.021022} {\bibfield  {journal}
		{\bibinfo  {journal} {Phys. Rev. X}\ }\textbf {\bibinfo {volume} {1}},\
		\bibinfo {pages} {021022} (\bibinfo {year} {2011})}\BibitemShut {NoStop}%
	\bibitem [{\citenamefont {Haine}\ and\ \citenamefont
		{Szigeti}(2015)}]{haine2015quantum}%
	\BibitemOpen
	\bibfield  {author} {\bibinfo {author} {\bibfnamefont {S.~A.}\ \bibnamefont
			{Haine}}\ and\ \bibinfo {author} {\bibfnamefont {S.~S.}\ \bibnamefont
			{Szigeti}},\ }\bibfield  {title} {\bibinfo {title} {Quantum metrology with
			mixed states: {When} recovering lost information is better than never losing
			it},\ }\href {https://doi.org/10.1103/PhysRevA.92.032317} {\bibfield
		{journal} {\bibinfo  {journal} {Phys. Rev. A}\ }\textbf {\bibinfo {volume}
			{92}},\ \bibinfo {pages} {032317} (\bibinfo {year} {2015})}\BibitemShut
	{NoStop}%
	\bibitem [{\citenamefont {Pham}\ \emph {et~al.}(2011)\citenamefont {Pham},
		\citenamefont {Sage}, \citenamefont {Stanwix}, \citenamefont {Yeung},
		\citenamefont {Glenn}, \citenamefont {Trifonov}, \citenamefont {Cappellaro},
		\citenamefont {Hemmer}, \citenamefont {Lukin}, \citenamefont {Park},
		\citenamefont {Yacoby},\ and\ \citenamefont {Walsworth}}]{pham2011magnetic}%
	\BibitemOpen
	\bibfield  {author} {\bibinfo {author} {\bibfnamefont {L.~M.}\ \bibnamefont
			{Pham}}, \bibinfo {author} {\bibfnamefont {D.~L.}\ \bibnamefont {Sage}},
		\bibinfo {author} {\bibfnamefont {P.~L.}\ \bibnamefont {Stanwix}}, \bibinfo
		{author} {\bibfnamefont {T.~K.}\ \bibnamefont {Yeung}}, \bibinfo {author}
		{\bibfnamefont {D.}~\bibnamefont {Glenn}}, \bibinfo {author} {\bibfnamefont
			{A.}~\bibnamefont {Trifonov}}, \bibinfo {author} {\bibfnamefont
			{P.}~\bibnamefont {Cappellaro}}, \bibinfo {author} {\bibfnamefont {P.~R.}\
			\bibnamefont {Hemmer}}, \bibinfo {author} {\bibfnamefont {M.~D.}\
			\bibnamefont {Lukin}}, \bibinfo {author} {\bibfnamefont {H.}~\bibnamefont
			{Park}}, \bibinfo {author} {\bibfnamefont {A.}~\bibnamefont {Yacoby}},\ and\
		\bibinfo {author} {\bibfnamefont {R.~L.}\ \bibnamefont {Walsworth}},\
	}\bibfield  {title} {\bibinfo {title} {Magnetic field imaging with
			nitrogen-vacancy ensembles},\ }\href
	{https://doi.org/10.1088/1367-2630/13/4/045021} {\bibfield  {journal}
		{\bibinfo  {journal} {New Journal of Physics}\ }\textbf {\bibinfo {volume}
			{13}},\ \bibinfo {pages} {045021} (\bibinfo {year} {2011})}\BibitemShut
	{NoStop}%
	\bibitem [{\citenamefont {Barry}\ \emph {et~al.}(2016)\citenamefont {Barry},
		\citenamefont {Turner}, \citenamefont {Schloss}, \citenamefont {Glenn},
		\citenamefont {Song}, \citenamefont {Lukin}, \citenamefont {Park},\ and\
		\citenamefont {Walsworth}}]{barry2016optical}%
	\BibitemOpen
	\bibfield  {author} {\bibinfo {author} {\bibfnamefont {J.~F.}\ \bibnamefont
			{Barry}}, \bibinfo {author} {\bibfnamefont {M.~J.}\ \bibnamefont {Turner}},
		\bibinfo {author} {\bibfnamefont {J.~M.}\ \bibnamefont {Schloss}}, \bibinfo
		{author} {\bibfnamefont {D.~R.}\ \bibnamefont {Glenn}}, \bibinfo {author}
		{\bibfnamefont {Y.}~\bibnamefont {Song}}, \bibinfo {author} {\bibfnamefont
			{M.~D.}\ \bibnamefont {Lukin}}, \bibinfo {author} {\bibfnamefont
			{H.}~\bibnamefont {Park}},\ and\ \bibinfo {author} {\bibfnamefont {R.~L.}\
			\bibnamefont {Walsworth}},\ }\bibfield  {title} {\bibinfo {title} {Optical
			magnetic detection of single-neuron action potentials using quantum defects
			in diamond},\ }\href {https://doi.org/10.1073/pnas.1601513113} {\bibfield
		{journal} {\bibinfo  {journal} {Proceedings of the National Academy of
				Sciences}\ }\textbf {\bibinfo {volume} {113}},\ \bibinfo {pages} {14133}
		(\bibinfo {year} {2016})}\BibitemShut {NoStop}%
	\bibitem [{\citenamefont {Savukov}\ and\ \citenamefont
		{Romalis}(2005)}]{savukov2005effects}%
	\BibitemOpen
	\bibfield  {author} {\bibinfo {author} {\bibfnamefont {I.~M.}\ \bibnamefont
			{Savukov}}\ and\ \bibinfo {author} {\bibfnamefont {M.~V.}\ \bibnamefont
			{Romalis}},\ }\bibfield  {title} {\bibinfo {title} {Effects of spin-exchange
			collisions in a high-density alkali-metal vapor in low magnetic fields},\
	}\href {https://doi.org/10.1103/PhysRevA.71.023405} {\bibfield  {journal}
		{\bibinfo  {journal} {Phys. Rev. A}\ }\textbf {\bibinfo {volume} {71}},\
		\bibinfo {pages} {023405} (\bibinfo {year} {2005})}\BibitemShut {NoStop}%
	\bibitem [{\citenamefont {Budker}\ and\ \citenamefont
		{Romalis}(2007)}]{budker2007optical}%
	\BibitemOpen
	\bibfield  {author} {\bibinfo {author} {\bibfnamefont {D.}~\bibnamefont
			{Budker}}\ and\ \bibinfo {author} {\bibfnamefont {M.}~\bibnamefont
			{Romalis}},\ }\bibfield  {title} {\bibinfo {title} {Optical magnetometry},\
	}\href {https://doi.org/10.1038/nphys566} {\bibfield  {journal} {\bibinfo
			{journal} {Nature physics}\ }\textbf {\bibinfo {volume} {3}},\ \bibinfo
		{pages} {227} (\bibinfo {year} {2007})}\BibitemShut {NoStop}%
	\bibitem [{\citenamefont {Appelt}\ \emph {et~al.}(1998)\citenamefont {Appelt},
		\citenamefont {Baranga}, \citenamefont {Erickson}, \citenamefont {Romalis},
		\citenamefont {Young},\ and\ \citenamefont {Happer}}]{appelt1998theory}%
	\BibitemOpen
	\bibfield  {author} {\bibinfo {author} {\bibfnamefont {S.}~\bibnamefont
			{Appelt}}, \bibinfo {author} {\bibfnamefont {A.~B.-A.}\ \bibnamefont
			{Baranga}}, \bibinfo {author} {\bibfnamefont {C.~J.}\ \bibnamefont
			{Erickson}}, \bibinfo {author} {\bibfnamefont {M.~V.}\ \bibnamefont
			{Romalis}}, \bibinfo {author} {\bibfnamefont {A.~R.}\ \bibnamefont {Young}},\
		and\ \bibinfo {author} {\bibfnamefont {W.}~\bibnamefont {Happer}},\
	}\bibfield  {title} {\bibinfo {title} {Theory of spin-exchange optical
			pumping of ${}^{3}\mathrm{He}$ and ${}^{129}\mathrm{Xe}$},\ }\href
	{https://doi.org/10.1103/PhysRevA.58.1412} {\bibfield  {journal} {\bibinfo
			{journal} {Phys. Rev. A}\ }\textbf {\bibinfo {volume} {58}},\ \bibinfo
		{pages} {1412} (\bibinfo {year} {1998})}\BibitemShut {NoStop}%
	\bibitem [{\citenamefont {Choi}\ \emph
		{et~al.}(2017{\natexlab{a}})\citenamefont {Choi}, \citenamefont {Choi},
		\citenamefont {Kucsko}, \citenamefont {Maurer}, \citenamefont {Shields},
		\citenamefont {Sumiya}, \citenamefont {Onoda}, \citenamefont {Isoya},
		\citenamefont {Demler}, \citenamefont {Jelezko}, \citenamefont {Yao},\ and\
		\citenamefont {Lukin}}]{choi2017depolarization}%
	\BibitemOpen
	\bibfield  {author} {\bibinfo {author} {\bibfnamefont {J.}~\bibnamefont
			{Choi}}, \bibinfo {author} {\bibfnamefont {S.}~\bibnamefont {Choi}}, \bibinfo
		{author} {\bibfnamefont {G.}~\bibnamefont {Kucsko}}, \bibinfo {author}
		{\bibfnamefont {P.~C.}\ \bibnamefont {Maurer}}, \bibinfo {author}
		{\bibfnamefont {B.~J.}\ \bibnamefont {Shields}}, \bibinfo {author}
		{\bibfnamefont {H.}~\bibnamefont {Sumiya}}, \bibinfo {author} {\bibfnamefont
			{S.}~\bibnamefont {Onoda}}, \bibinfo {author} {\bibfnamefont
			{J.}~\bibnamefont {Isoya}}, \bibinfo {author} {\bibfnamefont
			{E.}~\bibnamefont {Demler}}, \bibinfo {author} {\bibfnamefont
			{F.}~\bibnamefont {Jelezko}}, \bibinfo {author} {\bibfnamefont {N.~Y.}\
			\bibnamefont {Yao}},\ and\ \bibinfo {author} {\bibfnamefont {M.~D.}\
			\bibnamefont {Lukin}},\ }\bibfield  {title} {\bibinfo {title}
		{{Depolarization Dynamics in a Strongly Interacting Solid-State Spin
				Ensemble}},\ }\href {https://doi.org/10.1103/PhysRevLett.118.093601}
	{\bibfield  {journal} {\bibinfo  {journal} {Phys. Rev. Lett.}\ }\textbf
		{\bibinfo {volume} {118}},\ \bibinfo {pages} {093601} (\bibinfo {year}
		{2017}{\natexlab{a}})}\BibitemShut {NoStop}%
	\bibitem [{\citenamefont {Gershenfeld}\ and\ \citenamefont
		{Chuang}(1997)}]{gershenfeld1997bulk}%
	\BibitemOpen
	\bibfield  {author} {\bibinfo {author} {\bibfnamefont {N.~A.}\ \bibnamefont
			{Gershenfeld}}\ and\ \bibinfo {author} {\bibfnamefont {I.~L.}\ \bibnamefont
			{Chuang}},\ }\bibfield  {title} {\bibinfo {title} {{Bulk Spin-Resonance
				Quantum Computation}},\ }\href {https://doi.org/10.1126/science.275.5298.350}
	{\bibfield  {journal} {\bibinfo  {journal} {Science}\ }\textbf {\bibinfo
			{volume} {275}},\ \bibinfo {pages} {350} (\bibinfo {year}
		{1997})}\BibitemShut {NoStop}%
	\bibitem [{\citenamefont {Jones}\ \emph {et~al.}(2009)\citenamefont {Jones},
		\citenamefont {Karlen}, \citenamefont {Fitzsimons}, \citenamefont {Ardavan},
		\citenamefont {Benjamin}, \citenamefont {Briggs},\ and\ \citenamefont
		{Morton}}]{jones2009magnetic}%
	\BibitemOpen
	\bibfield  {author} {\bibinfo {author} {\bibfnamefont {J.~A.}\ \bibnamefont
			{Jones}}, \bibinfo {author} {\bibfnamefont {S.~D.}\ \bibnamefont {Karlen}},
		\bibinfo {author} {\bibfnamefont {J.}~\bibnamefont {Fitzsimons}}, \bibinfo
		{author} {\bibfnamefont {A.}~\bibnamefont {Ardavan}}, \bibinfo {author}
		{\bibfnamefont {S.~C.}\ \bibnamefont {Benjamin}}, \bibinfo {author}
		{\bibfnamefont {G.~A.~D.}\ \bibnamefont {Briggs}},\ and\ \bibinfo {author}
		{\bibfnamefont {J.~J.~L.}\ \bibnamefont {Morton}},\ }\bibfield  {title}
	{\bibinfo {title} {{Magnetic Field Sensing Beyond the Standard Quantum Limit
				Using 10-Spin NOON States}},\ }\href
	{https://doi.org/10.1126/science.1170730} {\bibfield  {journal} {\bibinfo
			{journal} {Science}\ }\textbf {\bibinfo {volume} {324}},\ \bibinfo {pages}
		{1166} (\bibinfo {year} {2009})}\BibitemShut {NoStop}%
	\bibitem [{\citenamefont {Simmons}\ \emph {et~al.}(2010)\citenamefont
		{Simmons}, \citenamefont {Jones}, \citenamefont {Karlen}, \citenamefont
		{Ardavan},\ and\ \citenamefont {Morton}}]{simmons2010magnetic}%
	\BibitemOpen
	\bibfield  {author} {\bibinfo {author} {\bibfnamefont {S.}~\bibnamefont
			{Simmons}}, \bibinfo {author} {\bibfnamefont {J.~A.}\ \bibnamefont {Jones}},
		\bibinfo {author} {\bibfnamefont {S.~D.}\ \bibnamefont {Karlen}}, \bibinfo
		{author} {\bibfnamefont {A.}~\bibnamefont {Ardavan}},\ and\ \bibinfo {author}
		{\bibfnamefont {J.~J.~L.}\ \bibnamefont {Morton}},\ }\bibfield  {title}
	{\bibinfo {title} {Magnetic field sensors using 13-spin cat states},\ }\href
	{https://doi.org/10.1103/PhysRevA.82.022330} {\bibfield  {journal} {\bibinfo
			{journal} {Phys. Rev. A}\ }\textbf {\bibinfo {volume} {82}},\ \bibinfo
		{pages} {022330} (\bibinfo {year} {2010})}\BibitemShut {NoStop}%
	\bibitem [{\citenamefont {Schaffry}\ \emph {et~al.}(2010)\citenamefont
		{Schaffry}, \citenamefont {Gauger}, \citenamefont {Morton}, \citenamefont
		{Fitzsimons}, \citenamefont {Benjamin},\ and\ \citenamefont
		{Lovett}}]{schaffry2010quantum}%
	\BibitemOpen
	\bibfield  {author} {\bibinfo {author} {\bibfnamefont {M.}~\bibnamefont
			{Schaffry}}, \bibinfo {author} {\bibfnamefont {E.~M.}\ \bibnamefont
			{Gauger}}, \bibinfo {author} {\bibfnamefont {J.~J.~L.}\ \bibnamefont
			{Morton}}, \bibinfo {author} {\bibfnamefont {J.}~\bibnamefont {Fitzsimons}},
		\bibinfo {author} {\bibfnamefont {S.~C.}\ \bibnamefont {Benjamin}},\ and\
		\bibinfo {author} {\bibfnamefont {B.~W.}\ \bibnamefont {Lovett}},\ }\bibfield
	{title} {\bibinfo {title} {Quantum metrology with molecular ensembles},\
	}\href {https://doi.org/10.1103/PhysRevA.82.042114} {\bibfield  {journal}
		{\bibinfo  {journal} {Phys. Rev. A}\ }\textbf {\bibinfo {volume} {82}},\
		\bibinfo {pages} {042114} (\bibinfo {year} {2010})}\BibitemShut {NoStop}%
	\bibitem [{\citenamefont {Boixo}\ \emph {et~al.}(2007)\citenamefont {Boixo},
		\citenamefont {Flammia}, \citenamefont {Caves},\ and\ \citenamefont
		{Geremia}}]{boixo_generalized_2007}%
	\BibitemOpen
	\bibfield  {author} {\bibinfo {author} {\bibfnamefont {S.}~\bibnamefont
			{Boixo}}, \bibinfo {author} {\bibfnamefont {S.~T.}\ \bibnamefont {Flammia}},
		\bibinfo {author} {\bibfnamefont {C.~M.}\ \bibnamefont {Caves}},\ and\
		\bibinfo {author} {\bibfnamefont {J.}~\bibnamefont {Geremia}},\ }\bibfield
	{title} {\bibinfo {title} {Generalized {Limits} for {Single}-{Parameter}
			{Quantum} {Estimation}},\ }\href
	{https://doi.org/10.1103/PhysRevLett.98.090401} {\bibfield  {journal}
		{\bibinfo  {journal} {Phys. Rev. Lett.}\ }\textbf {\bibinfo {volume} {98}},\
		\bibinfo {pages} {090401} (\bibinfo {year} {2007})}\BibitemShut {NoStop}%
	\bibitem [{\citenamefont {Pang}\ and\ \citenamefont
		{Brun}(2014)}]{pang_quantum_2014}%
	\BibitemOpen
	\bibfield  {author} {\bibinfo {author} {\bibfnamefont {S.}~\bibnamefont
			{Pang}}\ and\ \bibinfo {author} {\bibfnamefont {T.~A.}\ \bibnamefont
			{Brun}},\ }\bibfield  {title} {\bibinfo {title} {Quantum metrology for a
			general {Hamiltonian} parameter},\ }\href
	{https://doi.org/10.1103/PhysRevA.90.022117} {\bibfield  {journal} {\bibinfo
			{journal} {Phys. Rev. A}\ }\textbf {\bibinfo {volume} {90}},\ \bibinfo
		{pages} {022117} (\bibinfo {year} {2014})}\BibitemShut {NoStop}%
	\bibitem [{\citenamefont {Liu}\ \emph {et~al.}(2015)\citenamefont {Liu},
		\citenamefont {Jing},\ and\ \citenamefont {Wang}}]{liu2015quantum}%
	\BibitemOpen
	\bibfield  {author} {\bibinfo {author} {\bibfnamefont {J.}~\bibnamefont
			{Liu}}, \bibinfo {author} {\bibfnamefont {X.-X.}\ \bibnamefont {Jing}},\ and\
		\bibinfo {author} {\bibfnamefont {X.}~\bibnamefont {Wang}},\ }\bibfield
	{title} {\bibinfo {title} {Quantum metrology with unitary parametrization
			processes},\ }\href {https://doi.org/10.1038/srep08565} {\bibfield  {journal}
		{\bibinfo  {journal} {Scientific reports}\ }\textbf {\bibinfo {volume} {5}},\
		\bibinfo {pages} {8565} (\bibinfo {year} {2015})}\BibitemShut {NoStop}%
	\bibitem [{Note1()}]{Note1}%
	\BibitemOpen
	\bibinfo {note} {More generally, the eigenvectors of $\rho ^*$ may be written
		with a relative phase in the superpositions of $\mathinner {|{h_k}\delimiter
			"5365365 }$ ($\mathinner {|{\mu _k(0)}\delimiter "5365365 }$ in case of
		theorem \ref {th:time_dep}) such that the orthonormality condition remains
		fulfilled, i.e., $\protect \frac {\mathinner {|{h_k}\delimiter "5365365
			}+e^{\protect \ensuremath { \protect \mathrm {i\protect \tmspace +\thinmuskip
						{.1667em}} }\varphi _k}\mathinner {|{h_{d-k+1}}\delimiter "5365365
		}}{\protect \sqrt {2}}\protect \text { if } 2k<d+1$ and $\protect \frac
		{\mathinner {|{h_k}\delimiter "5365365 }-e^{\protect \ensuremath { \protect
					\mathrm {i\protect \tmspace +\thinmuskip {.1667em}} }\varphi _k}\mathinner
			{|{h_{d-k+1}}\delimiter "5365365 }}{\protect \sqrt {2}}\protect \text { if }
		2k>d+1$, where the same applies for theorem \ref {th:time_dep} when replacing
		the $\mathinner {|{h_k}\delimiter "5365365 }$ by $\mathinner {|{\mu
				_k(0)}\delimiter "5365365 }$.}\BibitemShut {Stop}%
	\bibitem [{\citenamefont {Bloomfield}\ and\ \citenamefont
		{Watson}(1975)}]{bloomfield1975inefficiency}%
	\BibitemOpen
	\bibfield  {author} {\bibinfo {author} {\bibfnamefont {P.}~\bibnamefont
			{Bloomfield}}\ and\ \bibinfo {author} {\bibfnamefont {G.~S.}\ \bibnamefont
			{Watson}},\ }\bibfield  {title} {\bibinfo {title} {The inefficiency of least
			squares},\ }\href {https://doi.org/10.1093/biomet/62.1.121} {\bibfield
		{journal} {\bibinfo  {journal} {Biometrika}\ }\textbf {\bibinfo {volume}
			{62}},\ \bibinfo {pages} {121} (\bibinfo {year} {1975})}\BibitemShut
	{NoStop}%
	\bibitem [{\citenamefont {Drury}\ \emph {et~al.}(2002)\citenamefont {Drury},
		\citenamefont {Liu}, \citenamefont {Lu}, \citenamefont {Puntanen},\ and\
		\citenamefont {Styan}}]{drury2002some}%
	\BibitemOpen
	\bibfield  {author} {\bibinfo {author} {\bibfnamefont {S.}~\bibnamefont
			{Drury}}, \bibinfo {author} {\bibfnamefont {S.}~\bibnamefont {Liu}}, \bibinfo
		{author} {\bibfnamefont {C.-Y.}\ \bibnamefont {Lu}}, \bibinfo {author}
		{\bibfnamefont {S.}~\bibnamefont {Puntanen}},\ and\ \bibinfo {author}
		{\bibfnamefont {G.~P.}\ \bibnamefont {Styan}},\ }\bibfield  {title} {\bibinfo
		{title} {{Some Comments on Several Matrix Inequalities with Applications to
				Canonical Correlations: Historical Background and Recent Developments}},\
	}\href@noop {} {\bibfield  {journal} {\bibinfo  {journal} {Sankhy{\=a}: The
				Indian Journal of Statistics, Series A}\ ,\ \bibinfo {pages} {453}} (\bibinfo
		{year} {2002})}\BibitemShut {NoStop}%
	\bibitem [{Note2()}]{Note2}%
	\BibitemOpen
	\bibinfo {note} {See Supplemental Material at [URL will be inserted by
		publisher] for proofs of theorem 1 and theorem 2, as well as the proofs of
		Heisenberg scaling for thermal states.}\BibitemShut {Stop}%
	\bibitem [{\citenamefont {Fra\"{i}sse}\ and\ \citenamefont
		{Braun}(2017)}]{fraisse_enhancing_2017}%
	\BibitemOpen
	\bibfield  {author} {\bibinfo {author} {\bibfnamefont {J.~M.~E.}\
			\bibnamefont {Fra\"{i}sse}}\ and\ \bibinfo {author} {\bibfnamefont
			{D.}~\bibnamefont {Braun}},\ }\bibfield  {title} {\bibinfo {title}
		{{Enhancing sensitivity in quantum metrology by {Hamiltonian} extensions}},\
	}\href {https://doi.org/10.1103/PhysRevA.95.062342} {\bibfield  {journal}
		{\bibinfo  {journal} {Phys. Rev. A}\ }\textbf {\bibinfo {volume} {95}},\
		\bibinfo {pages} {062342} (\bibinfo {year} {2017})}\BibitemShut {NoStop}%
	\bibitem [{\citenamefont {Pang}\ and\ \citenamefont
		{Jordan}(2017)}]{pang_optimal_2017}%
	\BibitemOpen
	\bibfield  {author} {\bibinfo {author} {\bibfnamefont {S.}~\bibnamefont
			{Pang}}\ and\ \bibinfo {author} {\bibfnamefont {A.~N.}\ \bibnamefont
			{Jordan}},\ }\bibfield  {title} {\bibinfo {title} {Optimal adaptive control
			for quantum metrology with time-dependent {Hamiltonians}},\ }\href
	{https://doi.org/10.1038/ncomms14695} {\bibfield  {journal} {\bibinfo
			{journal} {Nature Communications}\ }\textbf {\bibinfo {volume} {8}},\
		\bibinfo {pages} {14695} (\bibinfo {year} {2017})}\BibitemShut {NoStop}%
	\bibitem [{Note3()}]{Note3}%
	\BibitemOpen
	\bibinfo {note} {Yu showed that \protect \textit {any} state $\rho $ can be
		decomposed as $\DOTSB \tsum \slimits@ _k u_k\mathinner {|{U_k}\delimiter
			"5365365 }\mathinner {\delimiter "4360360 {U_k}|}$ with weights $u_k$ and
		(generally non-orthonormal) pure states $\mathinner {|{U_k}\delimiter
			"5365365 }$ such that the QFI equals $\DOTSB \tsum \slimits@ _k u_kI_\alpha
		(\mathinner {|{U_k}\delimiter "5365365 }\mathinner {\delimiter "4360360
			{U_k}|})$ \cite {Toth2013Extremal, yu2013quantum}. For $\rho ^*$ in
		Eq.~\protect \textup {\hbox {\mathsurround \z@ \protect \normalfont
				(\ignorespaces \ref {eq:opt_state}\unskip \@@italiccorr )}}, with rank $r\leq
		(d+1)/2$, Yu's state decomposition of $\rho ^*$ is given by the
		eigendecomposition of $\rho ^*$, $\mathinner {|{U_k}\delimiter "5365365
		}=\mathinner {|{\phi _k}\delimiter "5365365 }$ and $u_k=p_k$. This is not the
		case for $r> (d+1)/2$.}\BibitemShut {Stop}%
	\bibitem [{\citenamefont {Marshall}\ \emph {et~al.}(1979)\citenamefont
		{Marshall}, \citenamefont {Olkin},\ and\ \citenamefont
		{Arnold}}]{marshall1979inequalities}%
	\BibitemOpen
	\bibfield  {author} {\bibinfo {author} {\bibfnamefont {A.~W.}\ \bibnamefont
			{Marshall}}, \bibinfo {author} {\bibfnamefont {I.}~\bibnamefont {Olkin}},\
		and\ \bibinfo {author} {\bibfnamefont {B.~C.}\ \bibnamefont {Arnold}},\
	}\href {https://doi.org/10.1007/978-0-387-68276-1} {\emph {\bibinfo {title}
			{{Inequalities: Theory of Majorization and Its Applications}}}},\ Vol.\
	\bibinfo {volume} {143}\ (\bibinfo  {publisher} {Springer},\ \bibinfo {year}
	{1979})\BibitemShut {NoStop}%
	\bibitem [{\citenamefont {Fan}(1949)}]{fan1949theorem}%
	\BibitemOpen
	\bibfield  {author} {\bibinfo {author} {\bibfnamefont {K.}~\bibnamefont
			{Fan}},\ }\bibfield  {title} {\bibinfo {title} {On a theorem of {Weyl}
			concerning eigenvalues of linear transformations {I}},\ }\href
	{https://doi.org/10.1073/pnas.35.11.652} {\bibfield  {journal} {\bibinfo
			{journal} {Proceedings of the National Academy of Sciences of the United
				States of America}\ }\textbf {\bibinfo {volume} {35}},\ \bibinfo {pages}
		{652} (\bibinfo {year} {1949})}\BibitemShut {NoStop}%
	\bibitem [{\citenamefont {Fulton}(2000)}]{fulton2000eigenvalues}%
	\BibitemOpen
	\bibfield  {author} {\bibinfo {author} {\bibfnamefont {W.}~\bibnamefont
			{Fulton}},\ }\bibfield  {title} {\bibinfo {title} {Eigenvalues, invariant
			factors, highest weights, and {Schubert} calculus},\ }\href
	{https://doi.org/10.1090/S0273-0979-00-00865-X} {\bibfield  {journal}
		{\bibinfo  {journal} {Bulletin of the American Mathematical Society}\
		}\textbf {\bibinfo {volume} {37}},\ \bibinfo {pages} {209} (\bibinfo {year}
		{2000})}\BibitemShut {NoStop}%
	\bibitem [{\citenamefont {Schmitt}\ \emph {et~al.}(2017)\citenamefont
		{Schmitt}, \citenamefont {Gefen}, \citenamefont {St{\"u}rner}, \citenamefont
		{Unden}, \citenamefont {Wolff}, \citenamefont {M{\"u}ller}, \citenamefont
		{Scheuer}, \citenamefont {Naydenov}, \citenamefont {Markham}, \citenamefont
		{Pezzagna}, \citenamefont {Meijer}, \citenamefont {Schwarz}, \citenamefont
		{Plenio}, \citenamefont {Retzker}, \citenamefont {McGuinness},\ and\
		\citenamefont {Jelezko}}]{Schmitt832}%
	\BibitemOpen
	\bibfield  {author} {\bibinfo {author} {\bibfnamefont {S.}~\bibnamefont
			{Schmitt}}, \bibinfo {author} {\bibfnamefont {T.}~\bibnamefont {Gefen}},
		\bibinfo {author} {\bibfnamefont {F.~M.}\ \bibnamefont {St{\"u}rner}},
		\bibinfo {author} {\bibfnamefont {T.}~\bibnamefont {Unden}}, \bibinfo
		{author} {\bibfnamefont {G.}~\bibnamefont {Wolff}}, \bibinfo {author}
		{\bibfnamefont {C.}~\bibnamefont {M{\"u}ller}}, \bibinfo {author}
		{\bibfnamefont {J.}~\bibnamefont {Scheuer}}, \bibinfo {author} {\bibfnamefont
			{B.}~\bibnamefont {Naydenov}}, \bibinfo {author} {\bibfnamefont
			{M.}~\bibnamefont {Markham}}, \bibinfo {author} {\bibfnamefont
			{S.}~\bibnamefont {Pezzagna}}, \bibinfo {author} {\bibfnamefont
			{J.}~\bibnamefont {Meijer}}, \bibinfo {author} {\bibfnamefont
			{I.}~\bibnamefont {Schwarz}}, \bibinfo {author} {\bibfnamefont
			{M.}~\bibnamefont {Plenio}}, \bibinfo {author} {\bibfnamefont
			{A.}~\bibnamefont {Retzker}}, \bibinfo {author} {\bibfnamefont {L.~P.}\
			\bibnamefont {McGuinness}},\ and\ \bibinfo {author} {\bibfnamefont
			{F.}~\bibnamefont {Jelezko}},\ }\bibfield  {title} {\bibinfo {title}
		{{Submillihertz magnetic spectroscopy performed with a nanoscale quantum
				sensor}},\ }\href {https://doi.org/10.1126/science.aam5532} {\bibfield
		{journal} {\bibinfo  {journal} {Science}\ }\textbf {\bibinfo {volume}
			{356}},\ \bibinfo {pages} {832} (\bibinfo {year} {2017})}\BibitemShut
	{NoStop}%
	\bibitem [{Note4()}]{Note4}%
	\BibitemOpen
	\bibinfo {note} {The optimal control Hamiltonian of Pang et al.~\cite
		{pang_optimal_2017} fulfills exactly the condition in Eq.~\protect \textup
		{\hbox {\mathsurround \z@ \protect \normalfont (\ignorespaces \ref
				{eq:opt_con_th2}\unskip \@@italiccorr )}} which makes it optimal not only for
		pure but also for mixed states. For pure states, control Hamiltonians that
		ensure $U_\alpha (t)\mathinner {|{\mu _i(0)}\delimiter "5365365 }=\mathinner
		{|{\mu _i(t)}\delimiter "5365365 }$ only for $i=1,d$ would have been
		sufficient because only $\mu _1$ and $\mu _d$ contribute to the upper
		bound.}\BibitemShut {Stop}%
	\bibitem [{\citenamefont {Choi}\ \emph
		{et~al.}(2017{\natexlab{b}})\citenamefont {Choi}, \citenamefont {Yao},\ and\
		\citenamefont {Lukin}}]{choi2017quantum}%
	\BibitemOpen
	\bibfield  {author} {\bibinfo {author} {\bibfnamefont {S.}~\bibnamefont
			{Choi}}, \bibinfo {author} {\bibfnamefont {N.~Y.}\ \bibnamefont {Yao}},\ and\
		\bibinfo {author} {\bibfnamefont {M.~D.}\ \bibnamefont {Lukin}},\ }\bibfield
	{title} {\bibinfo {title} {Quantum metrology based on strongly correlated
			matter},\ }\href {https://arxiv.org/abs/1801.00042} {\bibfield  {journal}
		{\bibinfo  {journal} {arXiv preprint arXiv:1801.00042}\ } (\bibinfo {year}
		{2017}{\natexlab{b}})}\BibitemShut {NoStop}%
	\bibitem [{\citenamefont {De~Zanche}\ \emph {et~al.}(2008)\citenamefont
		{De~Zanche}, \citenamefont {Barmet}, \citenamefont {Nordmeyer-Massner},\ and\
		\citenamefont {Pruessmann}}]{de2008nmr}%
	\BibitemOpen
	\bibfield  {author} {\bibinfo {author} {\bibfnamefont {N.}~\bibnamefont
			{De~Zanche}}, \bibinfo {author} {\bibfnamefont {C.}~\bibnamefont {Barmet}},
		\bibinfo {author} {\bibfnamefont {J.~A.}\ \bibnamefont {Nordmeyer-Massner}},\
		and\ \bibinfo {author} {\bibfnamefont {K.~P.}\ \bibnamefont {Pruessmann}},\
	}\bibfield  {title} {\bibinfo {title} {{NMR} probes for measuring magnetic
			fields and field dynamics in {MR} systems},\ }\href
	{https://doi.org/10.1002/mrm.21624} {\bibfield  {journal} {\bibinfo
			{journal} {Magnetic Resonance in Medicine: An Official Journal of the
				International Society for Magnetic Resonance in Medicine}\ }\textbf {\bibinfo
			{volume} {60}},\ \bibinfo {pages} {176} (\bibinfo {year} {2008})}\BibitemShut
	{NoStop}%
	\bibitem [{\citenamefont {Scheuer}\ \emph {et~al.}(2017)\citenamefont
		{Scheuer}, \citenamefont {Schwartz}, \citenamefont {M\"uller}, \citenamefont
		{Chen}, \citenamefont {Dhand}, \citenamefont {Plenio}, \citenamefont
		{Naydenov},\ and\ \citenamefont {Jelezko}}]{scheuer2017robust}%
	\BibitemOpen
	\bibfield  {author} {\bibinfo {author} {\bibfnamefont {J.}~\bibnamefont
			{Scheuer}}, \bibinfo {author} {\bibfnamefont {I.}~\bibnamefont {Schwartz}},
		\bibinfo {author} {\bibfnamefont {S.}~\bibnamefont {M\"uller}}, \bibinfo
		{author} {\bibfnamefont {Q.}~\bibnamefont {Chen}}, \bibinfo {author}
		{\bibfnamefont {I.}~\bibnamefont {Dhand}}, \bibinfo {author} {\bibfnamefont
			{M.~B.}\ \bibnamefont {Plenio}}, \bibinfo {author} {\bibfnamefont
			{B.}~\bibnamefont {Naydenov}},\ and\ \bibinfo {author} {\bibfnamefont
			{F.}~\bibnamefont {Jelezko}},\ }\bibfield  {title} {\bibinfo {title} {Robust
			techniques for polarization and detection of nuclear spin ensembles},\ }\href
	{https://doi.org/10.1103/PhysRevB.96.174436} {\bibfield  {journal} {\bibinfo
			{journal} {Phys. Rev. B}\ }\textbf {\bibinfo {volume} {96}},\ \bibinfo
		{pages} {174436} (\bibinfo {year} {2017})}\BibitemShut {NoStop}%
	\bibitem [{\citenamefont {Pagliero}\ \emph {et~al.}(2014)\citenamefont
		{Pagliero}, \citenamefont {Laraoui}, \citenamefont {Henshaw},\ and\
		\citenamefont {Meriles}}]{pagliero_recursive_2014}%
	\BibitemOpen
	\bibfield  {author} {\bibinfo {author} {\bibfnamefont {D.}~\bibnamefont
			{Pagliero}}, \bibinfo {author} {\bibfnamefont {A.}~\bibnamefont {Laraoui}},
		\bibinfo {author} {\bibfnamefont {J.~D.}\ \bibnamefont {Henshaw}},\ and\
		\bibinfo {author} {\bibfnamefont {C.~A.}\ \bibnamefont {Meriles}},\
	}\bibfield  {title} {\bibinfo {title} {Recursive polarization of nuclear
			spins in diamond at arbitrary magnetic fields},\ }\href
	{https://doi.org/10.1063/1.4903799} {\bibfield  {journal} {\bibinfo
			{journal} {Applied Physics Letters}\ }\textbf {\bibinfo {volume} {105}},\
		\bibinfo {pages} {242402} (\bibinfo {year} {2014})}\BibitemShut {NoStop}%
	\bibitem [{\citenamefont {Tsang}\ \emph {et~al.}(2011)\citenamefont {Tsang},
		\citenamefont {Wiseman},\ and\ \citenamefont {Caves}}]{tsang2011fundamental}%
	\BibitemOpen
	\bibfield  {author} {\bibinfo {author} {\bibfnamefont {M.}~\bibnamefont
			{Tsang}}, \bibinfo {author} {\bibfnamefont {H.~M.}\ \bibnamefont {Wiseman}},\
		and\ \bibinfo {author} {\bibfnamefont {C.~M.}\ \bibnamefont {Caves}},\
	}\bibfield  {title} {\bibinfo {title} {{Fundamental Quantum Limit to Waveform
				Estimation}},\ }\href {https://doi.org/10.1103/PhysRevLett.106.090401}
	{\bibfield  {journal} {\bibinfo  {journal} {Phys. Rev. Lett.}\ }\textbf
		{\bibinfo {volume} {106}},\ \bibinfo {pages} {090401} (\bibinfo {year}
		{2011})}\BibitemShut {NoStop}%
	\bibitem [{\citenamefont {Berry}\ \emph {et~al.}(2015)\citenamefont {Berry},
		\citenamefont {Tsang}, \citenamefont {Hall},\ and\ \citenamefont
		{Wiseman}}]{berry2015quantum}%
	\BibitemOpen
	\bibfield  {author} {\bibinfo {author} {\bibfnamefont {D.~W.}\ \bibnamefont
			{Berry}}, \bibinfo {author} {\bibfnamefont {M.}~\bibnamefont {Tsang}},
		\bibinfo {author} {\bibfnamefont {M.~J.~W.}\ \bibnamefont {Hall}},\ and\
		\bibinfo {author} {\bibfnamefont {H.~M.}\ \bibnamefont {Wiseman}},\
	}\bibfield  {title} {\bibinfo {title} {{Quantum Bell-Ziv-Zakai Bounds and
				Heisenberg Limits for Waveform Estimation}},\ }\href
	{https://doi.org/10.1103/PhysRevX.5.031018} {\bibfield  {journal} {\bibinfo
			{journal} {Phys. Rev. X}\ }\textbf {\bibinfo {volume} {5}},\ \bibinfo {pages}
		{031018} (\bibinfo {year} {2015})}\BibitemShut {NoStop}%
	\bibitem [{\citenamefont {Uspensky}(1937)}]{uspensky1937introduction}%
	\BibitemOpen
	\bibfield  {author} {\bibinfo {author} {\bibfnamefont {J.~V.}\ \bibnamefont
			{Uspensky}},\ }\href@noop {} {\emph {\bibinfo {title} {{Introduction to
					Mathematical Probability}}}}\ (\bibinfo  {publisher} {McGraw-Hill Book
		Company, New York},\ \bibinfo {year} {1937})\BibitemShut {NoStop}%
	\bibitem [{\citenamefont {Holland}\ and\ \citenamefont
		{Burnett}(1993)}]{holland_interferometric_1993}%
	\BibitemOpen
	\bibfield  {author} {\bibinfo {author} {\bibfnamefont {M.~J.}\ \bibnamefont
			{Holland}}\ and\ \bibinfo {author} {\bibfnamefont {K.}~\bibnamefont
			{Burnett}},\ }\bibfield  {title} {\bibinfo {title} {Interferometric detection
			of optical phase shifts at the {Heisenberg} limit},\ }\href
	{https://doi.org/10.1103/PhysRevLett.71.1355} {\bibfield  {journal} {\bibinfo
			{journal} {Phys. Rev. Lett.}\ }\textbf {\bibinfo {volume} {71}},\ \bibinfo
		{pages} {1355} (\bibinfo {year} {1993})}\BibitemShut {NoStop}%
	\bibitem [{\citenamefont {Zwierz}\ \emph {et~al.}(2010)\citenamefont {Zwierz},
		\citenamefont {P\'erez-Delgado},\ and\ \citenamefont
		{Kok}}]{zwierz_general_2010}%
	\BibitemOpen
	\bibfield  {author} {\bibinfo {author} {\bibfnamefont {M.}~\bibnamefont
			{Zwierz}}, \bibinfo {author} {\bibfnamefont {C.~A.}\ \bibnamefont
			{P\'erez-Delgado}},\ and\ \bibinfo {author} {\bibfnamefont {P.}~\bibnamefont
			{Kok}},\ }\bibfield  {title} {\bibinfo {title} {General {Optimality} of the
			{Heisenberg} {Limit} for {Quantum} {Metrology}},\ }\href
	{https://doi.org/10.1103/PhysRevLett.105.180402} {\bibfield  {journal}
		{\bibinfo  {journal} {Phys. Rev. Lett.}\ }\textbf {\bibinfo {volume} {105}},\
		\bibinfo {pages} {180402} (\bibinfo {year} {2010})}\BibitemShut {NoStop}%
	\bibitem [{\citenamefont {Huelga}\ \emph {et~al.}(1997)\citenamefont {Huelga},
		\citenamefont {Macchiavello}, \citenamefont {Pellizzari}, \citenamefont
		{Ekert}, \citenamefont {Plenio},\ and\ \citenamefont {Cirac}}]{Huelga97}%
	\BibitemOpen
	\bibfield  {author} {\bibinfo {author} {\bibfnamefont {S.~F.}\ \bibnamefont
			{Huelga}}, \bibinfo {author} {\bibfnamefont {C.}~\bibnamefont
			{Macchiavello}}, \bibinfo {author} {\bibfnamefont {T.}~\bibnamefont
			{Pellizzari}}, \bibinfo {author} {\bibfnamefont {A.~K.}\ \bibnamefont
			{Ekert}}, \bibinfo {author} {\bibfnamefont {M.~B.}\ \bibnamefont {Plenio}},\
		and\ \bibinfo {author} {\bibfnamefont {J.~I.}\ \bibnamefont {Cirac}},\
	}\bibfield  {title} {\bibinfo {title} {{Improvement of Frequency Standards
				with Quantum Entanglement}},\ }\href
	{https://doi.org/10.1103/PhysRevLett.79.3865} {\bibfield  {journal} {\bibinfo
			{journal} {Phys. Rev. Lett.}\ }\textbf {\bibinfo {volume} {79}},\ \bibinfo
		{pages} {3865} (\bibinfo {year} {1997})}\BibitemShut {NoStop}%
	\bibitem [{\citenamefont {Kolodynski}(2014)}]{kolodynski2014precision}%
	\BibitemOpen
	\bibfield  {author} {\bibinfo {author} {\bibfnamefont {J.}~\bibnamefont
			{Kolodynski}},\ }\bibfield  {title} {\bibinfo {title} {Precision bounds in
			noisy quantum metrology},\ }\href {https://arxiv.org/abs/1409.0535}
	{\bibfield  {journal} {\bibinfo  {journal} {arXiv preprint arXiv:1409.0535}\
		} (\bibinfo {year} {2014})}\BibitemShut {NoStop}%
	\bibitem [{\citenamefont {Escher}\ \emph {et~al.}(2011)\citenamefont {Escher},
		\citenamefont {de~Matos~Filho},\ and\ \citenamefont
		{Davidovich}}]{escher_general_2011}%
	\BibitemOpen
	\bibfield  {author} {\bibinfo {author} {\bibfnamefont {B.~M.}\ \bibnamefont
			{Escher}}, \bibinfo {author} {\bibfnamefont {R.~L.}\ \bibnamefont
			{de~Matos~Filho}},\ and\ \bibinfo {author} {\bibfnamefont {L.}~\bibnamefont
			{Davidovich}},\ }\bibfield  {title} {\bibinfo {title} {General framework for
			estimating the ultimate precision limit in noisy quantum-enhanced
			metrology},\ }\href {https://doi.org/10.1038/nphys1958} {\bibfield  {journal}
		{\bibinfo  {journal} {Nat. Phys.}\ }\textbf {\bibinfo {volume} {7}},\
		\bibinfo {pages} {406} (\bibinfo {year} {2011})}\BibitemShut {NoStop}%
	\bibitem [{\citenamefont {T\'oth}\ and\ \citenamefont
		{Petz}(2013)}]{Toth2013Extremal}%
	\BibitemOpen
	\bibfield  {author} {\bibinfo {author} {\bibfnamefont {G.}~\bibnamefont
			{T\'oth}}\ and\ \bibinfo {author} {\bibfnamefont {D.}~\bibnamefont {Petz}},\
	}\bibfield  {title} {\bibinfo {title} {{Extremal properties of the variance
				and the quantum Fisher information}},\ }\href
	{https://doi.org/10.1103/PhysRevA.87.032324} {\bibfield  {journal} {\bibinfo
			{journal} {Phys. Rev. A}\ }\textbf {\bibinfo {volume} {87}},\ \bibinfo
		{pages} {032324} (\bibinfo {year} {2013})}\BibitemShut {NoStop}%
	\bibitem [{\citenamefont {Yu}(2013)}]{yu2013quantum}%
	\BibitemOpen
	\bibfield  {author} {\bibinfo {author} {\bibfnamefont {S.}~\bibnamefont
			{Yu}},\ }\bibfield  {title} {\bibinfo {title} {{Quantum Fisher Information as
				the Convex Roof of Variance}},\ }\href {https://arxiv.org/abs/1302.5311}
	{\bibfield  {journal} {\bibinfo  {journal} {arXiv preprint arXiv:1302.5311}\
		} (\bibinfo {year} {2013})}\BibitemShut {NoStop}%
\end{thebibliography}

\begin{thebibliography}{6}%
	\makeatletter
	\providecommand \@ifxundefined [1]{%
		\@ifx{#1\undefined}
	}%
	\providecommand \@ifnum [1]{%
		\ifnum #1\expandafter \@firstoftwo
		\else \expandafter \@secondoftwo
		\fi
	}%
	\providecommand \@ifx [1]{%
		\ifx #1\expandafter \@firstoftwo
		\else \expandafter \@secondoftwo
		\fi
	}%
	\providecommand \natexlab [1]{#1}%
	\providecommand \enquote  [1]{``#1''}%
	\providecommand \bibnamefont  [1]{#1}%
	\providecommand \bibfnamefont [1]{#1}%
	\providecommand \citenamefont [1]{#1}%
	\providecommand \href@noop [0]{\@secondoftwo}%
	\providecommand \href [0]{\begingroup \@sanitize@url \@href}%
	\providecommand \@href[1]{\@@startlink{#1}\@@href}%
	\providecommand \@@href[1]{\endgroup#1\@@endlink}%
	\providecommand \@sanitize@url [0]{\catcode `\\12\catcode `\$12\catcode
		`\&12\catcode `\#12\catcode `\^12\catcode `\_12\catcode `\%12\relax}%
	\providecommand \@@startlink[1]{}%
	\providecommand \@@endlink[0]{}%
	\providecommand \url  [0]{\begingroup\@sanitize@url \@url }%
	\providecommand \@url [1]{\endgroup\@href {#1}{\urlprefix }}%
	\providecommand \urlprefix  [0]{URL }%
	\providecommand \Eprint [0]{\href }%
	\providecommand \doibase [0]{https://doi.org/}%
	\providecommand \selectlanguage [0]{\@gobble}%
	\providecommand \bibinfo  [0]{\@secondoftwo}%
	\providecommand \bibfield  [0]{\@secondoftwo}%
	\providecommand \translation [1]{[#1]}%
	\providecommand \BibitemOpen [0]{}%
	\providecommand \bibitemStop [0]{}%
	\providecommand \bibitemNoStop [0]{.\EOS\space}%
	\providecommand \EOS [0]{\spacefactor3000\relax}%
	\providecommand \BibitemShut  [1]{\csname bibitem#1\endcsname}%
	\let\auto@bib@innerbib\@empty
	\bibitem [{\citenamefont {Bloomfield}\ and\ \citenamefont
		{Watson}(1975)}]{bloomfield1975inefficiency}%
	\BibitemOpen
	\bibfield  {author} {\bibinfo {author} {\bibfnamefont {P.}~\bibnamefont
			{Bloomfield}}\ and\ \bibinfo {author} {\bibfnamefont {G.~S.}\ \bibnamefont
			{Watson}},\ }\bibfield  {title} {\bibinfo {title} {The inefficiency of least
			squares},\ }\href {https://doi.org/10.1093/biomet/62.1.121} {\bibfield
		{journal} {\bibinfo  {journal} {Biometrika}\ }\textbf {\bibinfo {volume}
			{62}},\ \bibinfo {pages} {121} (\bibinfo {year} {1975})}\BibitemShut
	{NoStop}%
	\bibitem [{\citenamefont {Drury}\ \emph {et~al.}(2002)\citenamefont {Drury},
		\citenamefont {Liu}, \citenamefont {Lu}, \citenamefont {Puntanen},\ and\
		\citenamefont {Styan}}]{drury2002some}%
	\BibitemOpen
	\bibfield  {author} {\bibinfo {author} {\bibfnamefont {S.}~\bibnamefont
			{Drury}}, \bibinfo {author} {\bibfnamefont {S.}~\bibnamefont {Liu}}, \bibinfo
		{author} {\bibfnamefont {C.-Y.}\ \bibnamefont {Lu}}, \bibinfo {author}
		{\bibfnamefont {S.}~\bibnamefont {Puntanen}},\ and\ \bibinfo {author}
		{\bibfnamefont {G.~P.}\ \bibnamefont {Styan}},\ }\bibfield  {title} {\bibinfo
		{title} {{Some Comments on Several Matrix Inequalities with Applications to
				Canonical Correlations: Historical Background and Recent Developments}},\
	}\href@noop {} {\bibfield  {journal} {\bibinfo  {journal} {Sankhy{\=a}: The
				Indian Journal of Statistics, Series A}\ ,\ \bibinfo {pages} {453}} (\bibinfo
		{year} {2002})}\BibitemShut {NoStop}%
	\bibitem [{\citenamefont {Fulton}(2000)}]{fulton2000eigenvalues}%
	\BibitemOpen
	\bibfield  {author} {\bibinfo {author} {\bibfnamefont {W.}~\bibnamefont
			{Fulton}},\ }\bibfield  {title} {\bibinfo {title} {Eigenvalues, invariant
			factors, highest weights, and {Schubert} calculus},\ }\href
	{https://doi.org/10.1090/S0273-0979-00-00865-X} {\bibfield  {journal}
		{\bibinfo  {journal} {Bulletin of the American Mathematical Society}\
		}\textbf {\bibinfo {volume} {37}},\ \bibinfo {pages} {209} (\bibinfo {year}
		{2000})}\BibitemShut {NoStop}%
	\bibitem [{\citenamefont {Marshall}\ \emph {et~al.}(1979)\citenamefont
		{Marshall}, \citenamefont {Olkin},\ and\ \citenamefont
		{Arnold}}]{marshall1979inequalities}%
	\BibitemOpen
	\bibfield  {author} {\bibinfo {author} {\bibfnamefont {A.~W.}\ \bibnamefont
			{Marshall}}, \bibinfo {author} {\bibfnamefont {I.}~\bibnamefont {Olkin}},\
		and\ \bibinfo {author} {\bibfnamefont {B.~C.}\ \bibnamefont {Arnold}},\
	}\href {https://doi.org/10.1007/978-0-387-68276-1} {\emph {\bibinfo {title}
			{{Inequalities: Theory of Majorization and Its Applications}}}},\ Vol.\
	\bibinfo {volume} {143}\ (\bibinfo  {publisher} {Springer},\ \bibinfo {year}
	{1979})\BibitemShut {NoStop}%
	\bibitem [{\citenamefont {Pang}\ and\ \citenamefont
		{Jordan}(2017)}]{pang_optimal_2017}%
	\BibitemOpen
	\bibfield  {author} {\bibinfo {author} {\bibfnamefont {S.}~\bibnamefont
			{Pang}}\ and\ \bibinfo {author} {\bibfnamefont {A.~N.}\ \bibnamefont
			{Jordan}},\ }\bibfield  {title} {\bibinfo {title} {Optimal adaptive control
			for quantum metrology with time-dependent {Hamiltonians}},\ }\href
	{https://doi.org/10.1038/ncomms14695} {\bibfield  {journal} {\bibinfo
			{journal} {Nature Communications}\ }\textbf {\bibinfo {volume} {8}},\
		\bibinfo {pages} {14695} (\bibinfo {year} {2017})}\BibitemShut {NoStop}%
	\bibitem [{\citenamefont {Uspensky}(1937)}]{uspensky1937introduction}%
	\BibitemOpen
	\bibfield  {author} {\bibinfo {author} {\bibfnamefont {J.~V.}\ \bibnamefont
			{Uspensky}},\ }\href@noop {} {\emph {\bibinfo {title} {{Introduction to
					Mathematical Probability}}}}\ (\bibinfo  {publisher} {McGraw-Hill Book
		Company, New York},\ \bibinfo {year} {1937})\BibitemShut {NoStop}%
\end{thebibliography}
\end{document}